\documentclass[conference,compsoc]{IEEEtran}
\IEEEoverridecommandlockouts
\usepackage{cite}
\usepackage{amsmath,amssymb,amsfonts,amsthm}
\newtheorem{definition}{Definition}
\newtheorem{lemma}{Lemma}
\newtheorem{theorem}{Theorem}
\usepackage{algorithmic}
\usepackage{graphicx}
\usepackage{textcomp}
\usepackage{xcolor}
\usepackage{url}
\usepackage[hidelinks]{hyperref} 
\usepackage{subfigure}

\def\BibTeX{{\rm B\kern-.05em{\sc i\kern-.025em b}\kern-.08em
    T\kern-.1667em\lower.7ex\hbox{E}\kern-.125emX}}

\usepackage[inline]{enumitem}
\usepackage{listings}
\usepackage{float}
\usepackage{booktabs}
\usepackage{comment}

\usepackage{empheq}

\usepackage{bm}
\usepackage{url}
\usepackage{tikz-cd}
\usepackage[edges]{forest}
\usetikzlibrary{arrows.meta}
\usepackage[T1]{fontenc}
\usepackage{soul}

\definecolor{backcolour}{rgb}{0.95,0.95,0.92}
\definecolor{codegreen}{rgb}{0,0.6,0}
\definecolor{codepurple}{rgb}{0.58,0,0.82}

\pgfmathsetmacro{\redcomp}{11/255}
\pgfmathsetmacro{\greencomp}{96/255}
\pgfmathsetmacro{\bluecomp}{172/255}
\definecolor{codeblue}{rgb}{\redcomp, \greencomp, \bluecomp}

\lstdefinestyle{pnnstyle}{
    backgroundcolor=\color{backcolour},
    commentstyle=\color{codegreen},
    keywordstyle=\color{codeblue},
    stringstyle=\color{codepurple},
    basicstyle=\ttfamily\footnotesize,  % Code font size
    numberstyle=\tiny,  % Line number font size
    breakatwhitespace=false,
    xleftmargin=12pt,
    numbers=left,
    numbersep=4pt,  % Reduced space between line numbers and code
    frame=lines,
    rulecolor=\color{codeblue},
    breaklines=true,
    captionpos=b,
    keepspaces=true,
    showspaces=false,
    showstringspaces=false,
    showtabs=false,
    tabsize=2,
    mathescape=true,
    framexleftmargin=3pt,  % Adjust accordingly if needed
    escapechar=|,
    morekeywords={int64_t}  % Add more keywords here
}

\usepackage{array}
\newcolumntype{P}[1]{>{\raggedright\arraybackslash}p{#1}}
\usepackage{marvosym} % for \Lightning symbol

\usepackage{nicematrix}
\usepackage{wrapfig}
\usepackage{multirow}

\usepackage[linesnumbered,ruled,vlined]{algorithm2e}
\SetKwComment{Comment}{$\triangleright$\ }{}

\SetCommentSty{mycommfont}

\RestyleAlgo{ruled}
\DontPrintSemicolon
\makeatletter
\patchcmd{\algocf@makecaption@ruled}{\hsize}{\textwidth}{}{} % Caption to stretch full text width
\patchcmd{\@algocf@start}{-1.5em}{0em}{}{} % For // to right margin
\makeatother

\makeatletter
\newcommand{\removelatexerror}{\let\@latex@error\@gobble}
\makeatother

\RestyleAlgo{ruled}
\DontPrintSemicolon
\makeatletter
\patchcmd{\algocf@makecaption@ruled}{\hsize}{\textwidth}{}{} % Caption to stretch full text width
\patchcmd{\@algocf@start}{-1.5em}{0em}{}{} % For // to right margin
\makeatother

\usepackage[operators,sets,landau,probability]{cryptocode}

\newcommand\getsdollar{\mathrel{{\leftarrow}\vcenter{\hbox{\scriptsize\rmfamily\upshape\$}}}}
\newcommand{\red}[1]{\color{red!70!black}{#1}\color{black}}
\definecolor{darkgreen}{rgb}{0.0, 0.5, 0.0}

\newlist{questions}{enumerate}{1}
\setlist[questions]{label=\textbf{Question \arabic*}, leftmargin=*}

\definecolor{codegreen}{rgb}{0,0.6,0}
\definecolor{codegray}{rgb}{0.5,0.5,0.5}
\definecolor{codepurple}{rgb}{0.58,0,0.82}
\definecolor{backcolour}{rgb}{0.95,0.95,0.95}

\definecolor{main-color}{rgb}{0.6627, 0.7176, 0.7764}
\definecolor{back-color}{rgb}{0.1686, 0.1686, 0.1686}
\definecolor{string-color}{rgb}{0.3333, 0.5254, 0.345}
\definecolor{key-color}{rgb}{0.8, 0.47, 0.196}

\definecolor{light-gray}{gray}{0.95}

\usepackage[column=O]{cellspace}

\usepackage{etoolbox}

\newcounter{program}

\begin{document}

\title{Systematic Use of Random Self-Reducibility against Physical Attacks}

\author{
    \IEEEauthorblockN{
        Ferhat Erata\IEEEauthorrefmark{1},
        TingHung Chiu\IEEEauthorrefmark{2},
        Anthony Etim\IEEEauthorrefmark{1},
        Srilalith Nampally\IEEEauthorrefmark{2},
        Tejas Raju\IEEEauthorrefmark{2},
        Rajashree Ramu\IEEEauthorrefmark{2},\\
        Ruzica Piskac\IEEEauthorrefmark{1},
        Timos Antonopoulos\IEEEauthorrefmark{1},
        Wenjie Xiong\IEEEauthorrefmark{2}
        Jakub Szefer\IEEEauthorrefmark{1},
    }
    \IEEEauthorblockA{\IEEEauthorrefmark{1}Yale University, United States, \IEEEauthorrefmark{2}Virginia Tech, United States}
    % \IEEEauthorblockA{\IEEEauthorrefmark{2}Virginia Tech, United States}
}

\maketitle
\thispagestyle{plain} % TODO: remove page numbers

\begin{abstract}
  This work presents a novel, black-box software-based countermeasure against physical attacks including power side-channel and fault-injection attacks. The approach uses the concept of random self-reducibility and self-correctness to add randomness and redundancy in the execution for protection. Our approach is at the operation level, is not algorithm-specific, and thus, can be applied for protecting a wide range of algorithms. The countermeasure is empirically evaluated against attacks over operations like modular exponentiation, modular multiplication, polynomial multiplication, and number theoretic transforms. An end-to-end implementation of this countermeasure is demonstrated for RSA-CRT signature algorithm and Kyber Key Generation public key cryptosystems. The countermeasure reduced the power side-channel leakage by two orders of magnitude, to an acceptably secure level in TVLA analysis. For fault injection, the countermeasure reduces the number of faults to 95.4\% in average.
\end{abstract}


\section{Introduction}\label{sec:introduction}

Smart devices and IoT devices with sensors, processing capability, and actuators are becoming ubiquitous today in consumer electronics, healthcare, manufacturing, etc. These devices often collect sensitive or security-critical information and need to be protected. However, when deployed in the field, such devices are vulnerable to physical attackers who can have direct access to the devices.

Physical attacks can be categorized as passive attacks or active attacks. In passive attacks, such as \textit{Side-Channel Attacks (SCA)}, the attackers do not tamper with the execution, but can collect power traces, electromagnetic (EM) field traces, or traces of acoustic signals, and analyze the signals to learn information that is processed on the device. In active attacks, such as \textit{Fault Injection (FI) attacks}, the attackers can inject faults through a voltage glitch, clock glitch, EM field, or laser to cause a malfunction in the processing unit or memory to tamper with the execution to obtain desired results. It has been shown that both types of physical attacks have been able to break cryptography implementations to leak secret keys, for example~\cite{bozzato2019shaping,endo2011chip,schmidt2007optical,selmke2016attack}.

Even though the assumptions on the attacker's capability are similar for SCA and FI, the existing mitigation techniques treat the two types of attacks separately.  For side-channel attacks, the mitigation techniques usually use randomness or noise to decouple the signal observable by the attacker from the data value~\cite{kocher2011introduction, spreitzer2017systematic}. %please help this
For fault injection attacks, there are typically two solutions: one is attack detection and one is to have redundancy in the execution for error correction. The detection will detect when the execution has abnormal behavior, and then handle it as an exception. The error correction uses redundancy in the execution and uses the redundancy to correct execution error if there is~\cite{patrick2017lightweight}.
However, when we consider both SCA and FI attacks in the same system, separate mitigation for the two does not protect both attacks efficiently. For example, existing work~\cite{cojocar2018instruction} showed that instruction duplication as a fault tolerance mechanism amplifies the information leakage through side channels. Detection methods such as full, partial, encrypt-decrypt duplication \& comparison of a cipher~\cite{lomne2012need} produce repetitions of intermediate values that are exploitable by the side-channel adversary.

In this work, we propose a joint solution for both SCA and FI attacks. With a combination of random obfuscation using the Random Self-Reducibililty (RSR) property and redundancy for error correction, our proposed countermeasure is particularly effective against FI outperforming traditional redundancy-based methods. The randomness disrupts the attacker's observation of the statistics in fault attacks, thereby nullifying the effectiveness of statistical analysis as a tool for security compromise. This aspect is crucial in the face of increasingly sophisticated FI analysis techniques. In addition to its effectiveness against FI, the countermeasure also resists SCA, by rendering power consumption variations less useful to attackers. The countermeasure significantly enhances system security, particularly in environments where physical attacks are prevalent.

Another drawback of current mitigation techniques is that most existing work focuses on a certain implementation of a cryptographic algorithm, and to adopt the protection from one implementation to another needs redoing the security analysis process and redoing the implementation.

The proposed countermeasure offers significant benefits as a black box \textit{operation-level} solution {to both SCA and FI attacks}, and it is independent of the target algorithm being protected. This means there is no need for detailed knowledge of the implementation. \textit{The basis for the solution is to implement protection at low-level of operations such as modular exponentiation, modular multiplication, polynomial multiplication, and number theoretic transforms.} Also, we assume a generic fault model, and thus, there is no special fault profiling of a targeted device necessary.
Therefore, the proposed protection techniques can be applied directly in software without extensive system-specific adjustments.
In our evaluation, we showcase how the proposal protection techniques can be adopted to protect two different cryptosystems.

Our protection requires a small number of steps to implement. It can be implemented at C or high-level and is independent of the compiler or underlying architecture; assuming the compiler. First, target software is identified. Second, we locate low-level operations such as modular exponentiation, modular multiplication, polynomial multiplication, or number theoretic transforms. These operations can be protected with the idea of Random Self-Reducibility (RSR). Each instance of the low-level operation is replaced with an equivalent RSR operation. Each RSR operation requires querying a randomness source and then executing the low-level operations multiple times with original input values modified with the random values. Typically, multiple RSR operations are instantiated and majority voting is performed on the output of RSR operations. Because the protection works at the low-level operations such as modular exponentiation, modular multiplication, polynomial multiplication, or number theoretic transforms, it is independent of the higher-level algorithm or application. Since it does not rely on any hardware tricks, it is independent of the architecture and agnostic to the underlying compiler.

Our protection can be applied to any program or algorithm that uses modular exponentiation, modular multiplication, polynomial
multiplication, and number theoretic transforms to process secret or sensitive information. This encompasses major cryptogrpahic algorithms from ElGamal~\cite{elgamal1985public} and RSA~\cite{rivest1978method} to post-quantum cryptography such as Kyber~\cite{kyber} and Dilithium~\cite{ducas2018crystals}. In our evaluation, we show how our protection can be applied to RSA-CRT and Kyber's Key Generation algorithms.
Our contributions are summarized as follows:
\begin{itemize}[leftmargin=*]
  \item We propose a new software-based, combined countermeasure against power side-channel (Section~\ref{sec:randomization-against-power-side-channels}) and fault injection (Section~\ref{sec:self-correctness-against-fault-injections}) attacks, by randomizing the intermediate values of the computation using the notion of random self-reducibility (Section~\ref{sec:countermeasure}).
  \item We formalize the security of the countermeasure in relation to an attacker's fault injection capability, parameterize it, and quantify its effectiveness against fault-injection attacks, as detailed in Section~\ref{sec:security-guarantee}.
  \item End-to-end implementation of the countermeasure for RSA-CRT and Kyber's Key Generation public key cryptosystems (Section~\ref{sec:implementation}).
  \item Emprical evaluation of the countermeasure against power side-channel and fault-injection attacks over modular exponentiation, modular multiplication operations, polynomial multiplication, number theoretic transform operations, RSA-CRT, and Kyber's Key Generation (Section~\ref{sec:experiments}).
\end{itemize}

\section{Background}\label{sec:background}
In this section, we provide a brief overview on side-channels and fault injections as physical attacks.

\subsection{Power Side Channels}\label{sec:power-side-channels}

It is a well-known fact that the power consumption during certain stages of a cryptographic algorithm exhibits a strong correlation with the Hamming weight of its underlying variables, i.e., Hamming weight leakage model~\cite{platypus, CPA, mdpi:2020}.
This phenomenon has been widely exploited in the cryptographic literature in various attacks targeting a broad range of schemes, particularly post-quantum cryptographic implementations~\cite{2021:HOST, 2021:1307, 2020:1559, 2020:992, 2020:368, 2020:549, 2021:790, 2021:874, 2021:858, 2021:104, 2021:101}.
Therefore, we use the Hamming weight leakage model in the evaluation of the robustness of the countermeasure.

The Hamming weight leakage model assumes that the Hamming weight of the operands is strongly correlated with the power consumption. 
Each bit flip requires one or more voltage transitions from 0 to high (or vice versa). Different data values typically entail differing numbers of bit flips and therefore produce distinct power traces~\cite{power-analysis}.
Therefore, any circuit not explicitly designed to be resistant to power attacks has data-dependent power consumption. However, in a complex circuit, the differences can be so slight that they are difficult to distinguish from a {single trace}, particularly if an attacker's sampling rate is limited~\cite{platypus,hertzbleed}.
Therefore, it is necessary to use statistical techniques across multiple power traces~\cite{platypus}.

Test Vector Leakage Assessment (TVLA).~\cite{tvla} identifies if two sets of side channel measurements are {\itshape distinguishable} by computing the Welch's t-test for the two sets of measurements. It is being used in the literature to confirm the {\itshape presence} or {\itshape absence} of side-channel leakages for power traces, and has become the de facto standard in the evaluation of side-channel measurements~\cite{mdpi:2020, 2015:207, 2013:298, 2017:138, 2017:624}.
In side-channel analysis, the recommended thresholds for t-values are specifically tailored to detect potential information leakage in cryptographic systems. A t-value threshold of $\pm 4.5$ or $\pm 5$ is often considered in side-channel analysis. This threshold corresponds to a very high confidence level, rejecting the null hypothesis with a confidence greater than 99.999\% for a significantly large number of measurements. The null hypothesis typically being that all samples are drawn from the same distribution, a t-value outside this range indicates distinguishable distributions of the two sets and thus the existence of side-channel leakage \cite{socha2022comprehensive}. The choice of these thresholds is influenced by the need to balance the risk of false positives (incorrectly identifying information leakage when there is none) against the risk of false negatives (failing to detect actual information leakage). 

The Sum of Squared pairwise T-differences (SOST)~\cite{template} is a technique for identifying Points of Interest (PoIs) in side-channel analysis. It is particularly useful in scenarios where there are many data points (like traces in a cryptographic system), and you want to identify specific points in these traces that show significant variation based on different conditions or inputs.

\subsection{Fault Injection Attacks}\label{sec:fault-injection}
In the real world, there is a possibility that the devices will malfunction or be damaged, resulting in generating the error output, and we may ignore it. However, if the attacker intentionally induced the fault during the device operation, e.g., cryptographic calculation, he or she can recover the secret by analyzing the original and fault outputs. Most of the classical cryptographic algorithms can be attacked by fault injection attacks. For instance, the first fault attack research~\cite{boneh1997importance} was on the RSA implementation using the Chinese Remainder Theorem (CRT), which is the most common implementation of RAS used to secure communication. In this case, the attacker can recover the secret with only one faulty RSA-CRT signature. Moreover, in~\cite{mus2023jolt}, {Mus et al.} provide the fault attack method, which can attack El-Gamal or elliptic-curve (ECC) based signature, such as Schnorr signature and ECDSA, via Rowhammer (a software technique used to induce the fault in memory). Not only the public key cryptosystem but also the symmetric key cryptosystem are vulnerable to fault injection attacks. In~\cite{piret2003differential}, {Piret et al.} develop the fault attack method against substitution-permutation network (SPN) structures cryptographic algorithm, such as AES or KHAZAD. Even the post-quantum cryptographic algorithms~\cite{ravi2019exploiting}, which can protect against quantum computing, can be vulnerable to fault attacks. Therefore, it is necessary to have efficient FI attack protections that can be easily deployed.

The injected fault can, in principle, have an impact on any stage of the fetch-decode-execute cycle performed for each instruction~\cite{korak2014effects, zussa2013power}. Additionally, any optimizations implemented by the CPU, such as pipelining~\cite{yuce2015improving}, add to the complexity of executing a single instruction. Therefore, it is typically unknown what exactly goes wrong within the CPU when its behavior is changed due to fault injection, whereas the modified behavior itself is easier to measure. We consider a generic fault model, likely applicable to a wide range of targets, where a variable amount of bits in the instruction are flipped as a result of fault injection. Two types of behavior are possible using this fault model:
\begin{enumerate*}
  \item \textit{Instruction corruption}: the original instruction is modified into an instruction that has an impact on the behavior of the device. In practice, it may modify the instruction to any other instruction supported by the architecture.
  \item \textit{Instruction skipping}: effectively a subset of instruction corruption. The original instruction is corrupted into an instruction that does not have an impact on the behavior of the device. The resulting instruction does not change the execution flow or any state that is used later on.
\end{enumerate*}

Invocation of specific behavior is not a trivial task, as the low level control required to do this is often limited. However, it is possible to identify the more probable results while assuming that bit flips affecting single or all bits are more likely than complex patterns of bit flips~\cite{timmers2016controlling}.

On embedded processors, a fault model in which an attacker can skip an assembly instruction or equivalently replace it by a \texttt{nop} has been observed on several architectures and for several fault injection means~\cite{moro2014formal}.
Moro et al. in~\cite{moro2014experimental} assume that the effect of the injected fault on a 32-bit microcontroller leads to an instruction skip. Moro et al.~\cite{moro2014formal} and Barenghi et al.~\cite{barenghi2012fault} have proposed implementations of the \textit{Instruction Redundancy} technique as a countermeasure against this fault model.
Instruction skips correspond to specific cases of instruction replacements: replacing an instruction with another one that does not affect any useful register has the same effect as a \texttt{nop} replacement and so is equivalent to an instruction skip.

\section{Threat Model}\label{sec:threat-model}
In our threat model, we consider an attacker with physical access to a device, capable of injecting faults such as voltage glitches during the computation of a critical function like the number theoretic transform (see Section~\ref{sec:ntt}). These faults can corrupt or skip instructions (see Section~\ref{sec:fault-injection}) and happen anywhere multiple times but does not crash the device. Furthermore, the model permits the attacker to perform basic power side-channel analysis, collecting power trace samples. By correlating data-dependent power consumption with the Hamming weight leakage model, the attacker can expose vulnerabilities in cryptographic computations. This underscores the crucial need for robust defenses against both fault injection and side-channel attacks.

\section{Preliminaries}\label{sec:preliminaries}
We use the notion of {\it random self-reducibility}~\cite{blum1990self, rubinfeld1994robust} to develop a new software-based countermeasure against fault-injection attacks and simple power side-channel attacks. Therefore, in this section, we provide the necessary background on random self-reducibility. Since we apply our countermeasure to number-theoretic operations, we also provide the necessary background on number theoretic transforms.

\subsection{Notation}
The $\tilde{x}$ notation is used to represent a specific realization of a random variable (i.e., a specific value that the random variable takes on).
Let $\Pr_{x \in X}[\cdot]$ denote the probability of the event in the enclosed expression when $x$ is uniformly chosen from $X$.
We assume the domain and range of the function are the same set, usually named as $\mathbb{D}$, but the formalization can be expanded to accommodate multivariate functions and heterogeneous domains and ranges.

Let $q$ be a prime number, and the field of integers modulo $q$ be denoted as $\mathbb{Z}_q$. Schemes such as Kyber and Dilithium operate over polynomials in polynomial rings. The polynomial ring $\mathbb{Z}_q[x] / \phi(x)$ is denoted as $R_q$ where $\phi(x)=x^n+1$ is a cyclotomic polynomial with $n$ being a power of 2.
Multiplication of polynomials $a, b \in R_q$ is denoted as $c=a \cdot b \in R_q$.
Pointwise/Coefficientwise multiplication of two polynomials $a, b \in R_q$ is denoted as $c=a \circ b \in R_q$, which means that each of the coefficients of polynomial $a$ multiplies the coefficients of $b$ with the same index.
The NTT representation of a polynomial $a \in R_q$ is denoted as $\hat{a} \in R_q$.

\subsection{Random Self-Reducibility}\label{sec:random-self-reducibility}

Informally, a function $f$ is random-self-reducible if the evaluation of $f$ at any given instance $x$ can be reduced in polynomial time to the evaluation of $f$ at one or more random instances.

\begin{definition} [Random Self-Reducibililty (RSR)~\cite{blum1990self,rubinfeld1994robust}]\label{def:random-self-reducibility}
  Let $x \in \mathbb{D}$ and $c>1$ be an integer. We say that $f$ is $c$-random self-reducible if $f$ can be computed at any particular input $x$ via:
  \begin{equation}\label{eq:random-self-reducibility}
    F\left[f(x), f\left(a_1\right), \ldots, f\left(a_k\right), a_1, \ldots, a_k\right] = 0
  \end{equation}
  where $F$ can be computed asymptotically faster than $f$ and the $a_i$'s are uniformly distributed, although not necessarily independent; e.g., given the value of $a_1$ it is not necessary that $a_2$ be randomly distributed in $\mathbb{D}$. This notion of random self-reducibility is somewhat different than other definitions given by~\cite{abadi1987hiding, feigenbaum1993random, blum2019generate}, where the requirement on $F$ is that it be computable in polynomial time. %Equation~\ref{eq:random-self-reducibility} is a random self-reducible property of $f$.
\end{definition}

Another similar definition was made by Lipton~\cite{lipton1991new}. Suppose that we wish to compute the trivial identity function $f(x) = x$, and let $P$ be a program that computes $f(x)$. We can construct from $P$ another program $P'$
with the property
that it can compute $f(x)$ correctly at an arbitrary point $x$ provided that one can compute it at a number of random points.
Consider the following program $P'(x) \stackrel{\Delta}{=} \tilde{r}:=\text{random}(); \text{return } P(x+\tilde{r})-P(\tilde{r})$. $P'$ can compute $P(x)$ with inputs $x+\tilde{r}$ and $\tilde{r}$.

It is shown by Blum et al.~\cite{blum1990self} that self-correctors exist for any function that is \textit{random self-reducible}. A \textit{self-corrector} for $f$ takes a program $P$ that is correct on most inputs and turns it into a program that is correct on every input with high probability.

\subsection{Arithmetic Secret Sharing}

Privacy-preserving computing allows multiple parties to evaluate a function while keeping the inputs private and revealing only the output of the function and nothing else.  One popular approach to outsourcing sensitive workloads to untrusted workers is to use arithmetic secret sharing~\cite{damgaard2012multiparty,mohassel2018aby3}. It splits a secret into multiple shares, distributing them across various workers. Each worker processes their respective share locally. Assuming the workers will not collude, it is information-theoretically impossible for each worker to recover the secret from its share~\cite{xiong2022secndp}.

In standard arithmetic secret sharing, the client aims to compute \( f(x, y) = ax + by = z \), with the property that \( f(x,y) = f(x_1, y_1) + f(x_2, y_2) \), where \( x_1 \) and \( y_1 \) are randomly chosen. 
Let $\mathbb{Z}\left(2^{w_e}\right)$ denote the integer ring of size $2^{w_e}$. The shares are constructed such that the sum of all shares is equal to the original secret value $x \in$ $\mathbb{Z}\left(2^{w_e}\right)$.
The client then delegates the computation to workers (untrusted entities). These workers independently calculate \( f(x_1, y_1) \) and \( f(x_2, y_2) \), then relay their results back to the client. The client derives \( f(x, y) \) using these partial results from the untrusted workers. The randomness in the shares is crucial for our power side-channel countermeasure. While arithmetic secret sharing is based on the linearity of addition and multiplication over integers, our approach utilizes Random Self-Reducible properties, some of which may not necessarily be linear. Moreover, to counteract fault injection attacks, our algorithm must produce accurate results despite faults. We achieve this in our countermeasure by repeating the computation \( n \) times and choosing the majority of the responses.

\section{Overview of Our Countermeasure}\label{sec:countermeasure}

The foundational works of Blum et al.~\cite{blum1990self} and Lipton~\cite{lipton1991new} on testing have significantly influenced our approach to developing countermeasures. We have incorporated the concept of self-correctness to safeguard against fault-injection attacks, and the principles of random self-reducibility and randomly-testable functions to defend against power side-channel attacks.
These notions are investigated and applied as a countermeasure against physical attacks in the literature.

At the heart of this method is the generic, randomized Algorithm~\ref{alg:psca-countermeasure}, which is founded on the principle described in Definition~\ref{def:random-self-reducibility}. Additionally, Algorithm~\ref{alg:fia-countermeasure} boosts the effectiveness of the randomized Algorithm~\ref{alg:psca-countermeasure} through majority voting and probability amplification~\cite{sinclair2011randomness}.

We observed that instance hiding can be also used against physical attacks, such as power side-channel and fault-injection attacks
by randomizing the intermediate values of the computation. In this way, attackers won't be able to correlate the side-channel leakage with the intermediate values of the computation (see~\ref{fig:countermeasure-motivation}).
For example, secrets in ElGamal Decryption~\cite{elgamal1985public} (see Algorithm~\ref{alg:elgamal-decryption}) can be protected end-to-end using instance hiding, but instead of using arithmetic secret sharing, we use random self-reducible properties.

\begin{figure*}
  \centering
  \includegraphics[width=.9\textwidth]{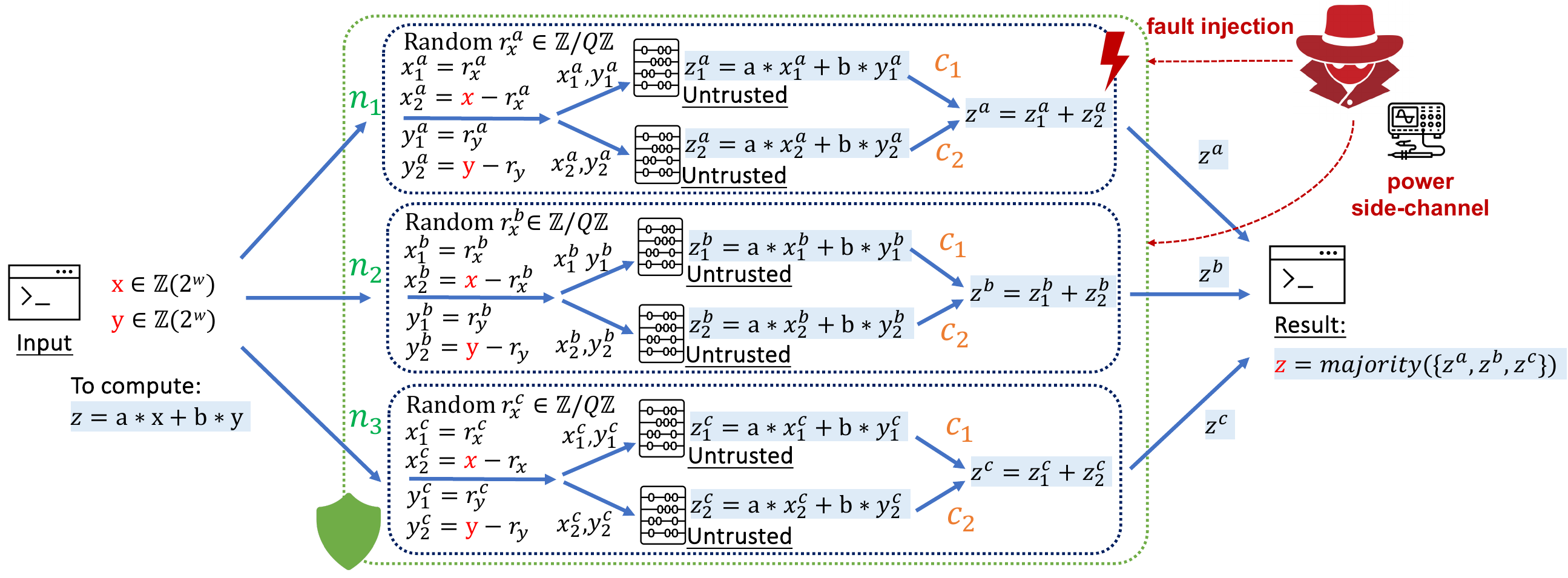}
  % \caption{Motivation: In regular arithmetic secret sharing, the client wants to compute $f(x, y) = ax + by = z$ whose property is $f(x,y) = f(x_1, y_1) + f(x_2, y_2)$ where $x_1$ and $x_2$ are selected randomly ($c_1, c_2$). Client delegates the computation to the workers (untrusted computation). The workers compute $f(x_1, y_1)$ and $f(x_2, y_2)$ locally and send the results to the client. The client computes $f(x, y)$ from the partial results sent by untrusted workers.  This randomness in the shares is the key to our power side-channel countermeasure. Arithmetic secret sharing relies on the linearity of addition and multiplication over integers, $f(x+y)=f(x)+f(y)$. However, we rely on Random Self-Reducible properties, some of which are not nessearly linear. In addition, against fault injection attacks, the algorithm should return correct result under faults. In our countermeasure, we repeat the computation $n$ times and select the majority of the answers.  }
  % We can increase the number of shares to distribute workloads and thereby increasing security such that for three workers the property becomes $f(x,y) = f(x_1, y_1) + f(x_2, y_2) + f(x_3, y_3)$ where $x_1 + x_2 + x_3 = x$ and $y_1 + y_2 + y_3 = y$. 
\caption{
    Motivation: In standard arithmetic secret sharing, the client aims to compute \( f(x, y) = ax + by = z \), with the property that \( f(x,y) = f(x_1, y_1) + f(x_2, y_2) \), where \( x_1 \) and \( x_2 \) are randomly chosen (annotated by \( c_1, c_2 \)). The client delegates the computation to workers (untrusted entities). These workers independently calculate \( f(x_1, y_1) \) and \( f(x_2, y_2) \), then relay their results back to the client. The client derives \( f(x, y) \) using these partial results from the untrusted workers. The randomness in the shares is crucial for our power side-channel countermeasure. While arithmetic secret sharing is based on the linearity of addition and multiplication over integers (i.e., \( f(x+y)=f(x)+f(y) \)), our approach utilizes Random Self-Reducible properties, some of which may not necessarily be linear. Moreover, to counteract fault injection attacks, our algorithm must produce accurate results despite faults. We achieve this in our countermeasure by repeating the computation \( n \) times and choosing the majority of the responses.
}

  \label{fig:countermeasure-motivation}
\end{figure*}

\begin{algorithm}[ht]
  \caption{ElGamal Decryption}\label{alg:elgamal-decryption}
  \small
  \LinesNumbered  % Enable line numbers
  \SetKwInOut{Input}{Input}
  \SetKwInOut{Output}{Output}
  \Input{Ciphertexts: $c_1$, $c_2$, Secret Key: $x$}
  \Output{Decrypted Message: $m$}
  % \BlankLine
  Calculate $s := c_1^x \bmod R$\;
  Calculate $l := s^{-1} \bmod R$\;
  Calculate $m := c_2 \cdot l \bmod R$\;
  \Return $m$
\end{algorithm}

In this algorithm, $c_1$ and $c_2$ are the ciphertexts, $x$ is the secret key, and $m$ is the decrypted message. The operation $c_1^x$ represents raising $c_1$ to the power of $x$, and $s^{-1}$ represents the modular multiplicative inverse of $s$ (i.e., $s^{(-1)} \bmod{R}$, where $R$ is the prime modulus used in the ElGamal encryption scheme). The result of the decryption, $m$, is obtained by multiplying $c_2$ with the modular inverse of $s$, denoted as $l$ in this algorithm.

For instance, we can protect the modular exponentiation function, {\small$ f(a, x, R) = a^x \bmod R$}, in ElGamal decryption using {\small $P(a, x, R) = P(a, \tilde{x}_1, R) \cdot_R P(a, \tilde{x}_2, R)$}, and modular multiplication, {\small $f(x, y, R)=x \cdot_R y$}, using {\small $P(x, y, R) =  P(\tilde{x}_1, \tilde{y}_1, R) +_R P(\tilde{x}_2, \tilde{y}_1, R) +_R P(\tilde{x}_1, \tilde{y}_2, R) +_R P(\tilde{x}_2, \tilde{y}_2, R)$}. In these equalities shares should be selected to make {\small $x = x_1 +_R y_1$} and {\small $y = y_1 +_R y_2$}.

\begin{table}[ht]
  \caption{Operations used in ElGamal Decryption: Deshpande et al.~\cite{deshpande2021modular} use two different ways to implement modular exponentation ($\mathrm{D_1}$ and $\mathrm{D_1}$) in ElGamal decryption.
  }\label{tab:elgamal}
  \centering
  \small
  \begin{NiceTabular}{|l|l|l|}
    \midrule
    \multirow{1}{*}{} & \multirow{1}{*}{\bf Method-$\mathbf{D_1}$} & \multirow{1}{*}{\bf Method-$\mathbf{D_2}$}  \\\midrule
    $s:=c_1^x$        & Modular exponentiation                     & Modular exponentiation                      \\
    $l:=s^{-1}$       & Fermat's method~\cite{weil2013basic}       & Fast GCD algorithm~\cite{bernstein2019fast} \\
    $m:=c_2.l$        & Modular multiplication                     & Modular multiplication                      \\
    \midrule
  \end{NiceTabular}
\end{table}

It is also hard to implement countermeasures for different implementations of the same mathematical function.
Table~\ref{tab:elgamal} presents two distinct methodologies for implementing the ElGamal decryption algorithm. Method $D_1$ employs Fermat's method (Algorithm~\ref{alg:modular-exponentiation} in Appendix) for the modular inverse calculation~\cite{weil2013basic}, while $D_2$ utilizes a sophisticated, constant-time modular inverse implementation recently developed by Bernstein and Yang~\cite{bernstein2019fast} (Algorithm~\ref{alg:fast-gcd} in Appendix). Importantly, our countermeasure technique, applicable for both modular exponentiation (refer to Section~\ref{sec:modular-exponentiation-1}) and modular multiplication (refer to Section~\ref{sec:modular-multiplication-1}), is compatible with either method regardless of the complexity of their respective implementations.

\subsection{RSR against Power Side Channels}\label{sec:randomization-against-power-side-channels}

% \hl{Can you say "RSR Against Power Side Channels"? "Randomization" seems very generic}
Consider a correct program $P$ that has an associated random self-reducible property, which takes the form of a functional equation $p$. This property is deemed satisfied if, in the equation $p$, we can substitute $P$ for the function $f$ and the equation remains true.

\begin{algorithm}
  \caption{$c$-secure-countermeasure PSCA $(P, x, c)$.}\label{alg:psca-countermeasure}
  \small
  \LinesNumbered  % Enable line numbers

  \SetKwInOut{Input}{Input}
  \SetKwInOut{Output}{Output}
  \Input{Program: $P$, Sensitive input: $x$, Security: $c$}
  \Output{$P(x)$}

  Randomly split $a_1, \ldots, a_c$ based on $x$. \;
  \For{$i=1, \ldots, c$}{
    $\alpha_i \gets P(a_i)$
  }
  \Return $F[x, a_1, \ldots, a_c, \alpha_1, \ldots, \alpha_c]$
\end{algorithm}

Generic $c$-secure-countermeasure PSCA $(P, x, c)$ defined Algorithm~\ref{alg:psca-countermeasure} takes a program $P$, a sensitive input $x$, and a security parameter $c$. The algorithm randomly splits $x$ into $c$ shares $a_1, \ldots, a_c$ such that $x=a_1+\cdots+a_c$, and calls $P$ on each share $a_i$ to obtain $\alpha_i = P(a_i)$. Finally, the algorithm returns the result of the function $F$ on $x, a_1, \ldots, a_c, \alpha_1, \ldots, \alpha_c$. The function basis $F$ is defined based on the random self-reducible property of the function $f$ that $P$ implements (cf. Definition~\ref{def:random-self-reducibility}).

To ensure minimum security, splitting the secret input into two shares would suffice. However, for enhanced security, the secret input can be divided into additional shares. It's important to view the security parameter $c$ as an invocation to $P$, especially in the context of bivariate functions, rather than merely the number of shares.

\vspace*{1em}
\noindent
{\centering\fbox{%
    \parbox{.96\columnwidth}{%
      \textbf{Masking with Random Self-Reducibility}. If a cryptographic operation has a random self-reducible property, then it is possible to protect it against power side-channel attacks by masking with arithmetic secret sharing.
    }%
  }
}
\vskip2pt

\subsection{Self-Correctness against Fault Injections}\label{sec:self-correctness-against-fault-injections}

Fault injection attacks rely on obtaining a faulty output or correlating the faulty output with the input or secret-dependent intermediate values.
By introducing redundancy and majority voting, we can obtain correct results even if some results are incorrect due to injected faults.

In Algorithm~\ref{alg:fia-countermeasure}, we show how to apply the fault injection countermeasure approach on top of the power side-channel countermeasure. To protect a program $P$ that implements a function $f$ having a random self-reducible property, the algorithm calls $P$'s $c$-secure-countermeasure $n$ times and returns the majority of the answers. The function $c$-secure-countermeasure takes a program $P$, a sensitive input $x$, and a security parameter $c$.

\begin{algorithm}[h]
  \caption{$n$-secure countermeasure FIA $(P, x, n, c)$.}\label{alg:fia-countermeasure}
  \small
  \LinesNumbered

  \SetKwInOut{Input}{Input}
  \SetKwInOut{Output}{Output}
  \Input{Program: $P$, Sensitive input: $x$, Security: $n, c$}
  \Output{$P(x)$}

  % \BlankLine
  \For{$m=1, \ldots, n$}{
    $\text{answer}_m \gets$ call $c$-secure-countermeasure($P, x, c$)
  }
  \Return the majority in $\{\text{answer}_m\text{: } m=1, \ldots, n\}$
\end{algorithm}

Note that \( c \) and \( n \) are independent security parameters. The security parameter \( c \) represents the number of calls to the unprotected program used in the PSCA countermeasure, whereas \( n \) signifies the number of iterations in the FIA countermeasure. The security parameter \( n \) is associated with the attacker's capability to inject effective faults. Owing to redundancy, an increase in the security parameter \( c \) results in a decreased likelihood of the attacker successfully injecting a fault.

\begin{algorithm}[h]
  \caption{($c$, $n$)-secure mod operation $(P, R, x, c, n)$.}\label{alg:mod-operation}
  \small
  \setlength{\algomargin}{2pt}
  \LinesNumbered
  \SetKwInOut{Input}{Input}
  \SetKwInOut{Output}{Output}
  \Input{Program: $P$, Sensitive input: $x$, Security: $n, c$}
  \Output{$P(x)$}

  % \BlankLine
  \For{$m=1, \ldots, n$}{
    $x_1, x_2, \ldots, x_c$ $\getsdollar $ Random-Split($R 2^n, x$) \\
    $\text{answer}_m \gets P(x_1, R) +_R P(x_2, R) \ldots +_R P(x_c, R)$
  }
  \Return the majority in $\{\text{answer}_m\text{: } m=1, \ldots, n\}$
\end{algorithm}

Algorithm~\ref{alg:mod-operation} presents an example of a combined and configurable countermeasure, effective against both PSCA and FIA, applied to the modular multiplication operation. In Line 2, the algorithm divides the input $x$ into $c$ shares $x_1, x_2, \ldots, x_c$, satisfying $x = x_1 + x_2 + \cdots + x_c$. The methodology for the random splitter algorithm is detailed in Section~\ref{sec:random-split-function}. Furthermore, the majority function, which essentially returns the most common answer, is described in Section~\ref{sec:majority-function}.

\vspace*{1em}
\noindent
{\centering\fbox{%
    \parbox{.96\columnwidth}{%
      \textbf{Self-Correctness with Majority Voting}. Fault injection attacks rely on faulty output. By majority voting, we can obtain correct results even if some results are incorrect.
    }%
  }
}
\vskip2pt

\subsection{n and attacker's probability of success}\label{sec:security-guarantee}

Fault injection occurs at the hardware level and is both challenging and unpredictable to control. When a successful fault is induced, it transforms a previously correct victim program into an incorrect one. Consequently, the essence of a fault injection attack is its probabilistic nature. This concept is abstracted in terms of the attacker's probability of success in our work.

\begin{definition}[$\varepsilon$-fault tolerance]\label{def:epsilon-fault-tolerance}
  % Let $\mathbb{T}$ be the attacker's total number of trials and 
  Let $\varepsilon$ be the upper bound on the attacker's probability of injecting a fault successfully at an unprotected program $P$ that correctly implements a function $f$. Say that the program $P$ is $\varepsilon$-fault tolerant for the function $f$ provided $P(x)=f(x)$ for at least $1-\varepsilon$ of any input $x$. We assume each fault injection is independent of the others:
  $\operatorname{Pr}_{fault}[P(x) \neq f(x)]<\varepsilon.$
\end{definition}

Algorithm~\ref{alg:psca-countermeasure} is a randomized algorithm and Algorithm~\ref{alg:fia-countermeasure} is also a randomized algorithm that repeats the computation $n$ times by calling Algorithm~\ref{alg:psca-countermeasure} and uses majority voting to pick the correct answer. Therefore, we can use Chernoff bounds~\cite{sinclair2011randomness} to show that the probability of getting the correct answer is at least $1-\delta$.

A simple and common use of Chernoff bounds is for "boosting" of randomized algorithms.
If one has an algorithm that outputs a guess that is the desired answer with probability $p>1/2$, then one can get a higher success rate by running the algorithm $n=\log (1/ \delta) 2 p /(p-1 / 2)^2$ times and outputting a guess that is output by more than $n/2$ runs of the algorithm. Assuming that these algorithm runs are independent, the probability that more than $n/2$ of the guesses is correct is equal to the probability that the sum of independent Bernoulli random variables $X_k$ that are 1 with probability $p$ is more than $n/2$. This can be shown to be at least $1-\delta$ via the multiplicative Chernoff bound ($\mu=n p$)~\cite{wiki:chernoff}: $\Pr\left[X>{n}/{2}\right] \geq 1-e^{-n(p-1 / 2)^2 /(2 p)} \geq 1-\delta.$

\begin{theorem}[Derived from Theorem 3.1 in \cite{lipton1991new}]\label{thm:1}
  Suppose that $f$ is randomly self-reducible and that $P$ is $\varepsilon$-fault tolerant for the function $f$. Consider a $c$-secure countermeasure $\widetilde{C}(x)$ (Line 4 in Algorithm~\ref{alg:psca-countermeasure}):
  $$
    \textbf{return } F[x, a_1, \ldots, a_c, P(a_1), \ldots, P(a_c)]
  $$
  Then, for any $x, \widetilde{C}(x)$ is equal to $f(x)$ with probability at least $1-\varepsilon c$.
\end{theorem}

\begin{proof}
  Fix an input $x$. Clearly, the probability that $\widetilde{C}(x)$ is correct is at least the probability that for each $i, P(a_i)=f(a_i)$. This follows since $f$ is random self-reducible with respect to the number of calls to $P$ is done. It therefore follows that $\widetilde{C}$ returns correct results at least $1-\varepsilon c$ of the time.
\end{proof}

In the next sections, we will present a number of examples of $c$-secure countermeasures whose security parameter is mostly $c=2$. Thus, for these functions, Theorem~\ref{thm:1} says that, for $\varepsilon$ equal to $1/100$, the probability that $\widetilde{C}$ returns correct results is at least 0.98. We can amplify the probability of success by repeating the computation $n$ times and using majority voting. In addition, we can select a bigger $n$ by adjusting $\delta$ as the confidence parameter:

\vspace*{1em}
\noindent
{\centering\fbox{%
    \parbox{.96\columnwidth}{%
      \textbf{Lower bound for $\bf{n}$}. The attacker's probability of success is $\varepsilon$, and for a $c$-secure countermeasure, the lower bound for $n$ is defined as:
      $n=\log (1 / \delta) 2 (1-\varepsilon c) /(\varepsilon c/2)^2$,
      where $\delta$ is the confidence parameter.
    }%
  }
}
\vskip4pt

Algorithm~\ref{alg:psca-countermeasure} makes calls to a program $P$ that implements a function $f$ having a random self-reducible property. However, we do not need to know the implementation of the function $f$, we just need to know the mathematical definition of the function $f$ to configure the Algorithm~\ref{alg:psca-countermeasure} and~\ref{alg:fia-countermeasure}.
Therefore, one further advantage of our countermeasure is that it follows ``black-box'' approach. The fault injection attacks are hardware attacks, and the attacker does not have access to the software implementation of the function. Therefore, the attacker can only observe the input and output of the function. By using the black-box approach, we basically make the countermeasure robust at the hardware level.

\vspace*{1em}
\noindent
{\centering\fbox{%
    \parbox{0.96\columnwidth}{%
      \textbf{Black-box}. If we replace the $f$ function with a program $P$ that computes the function $f$, then our countermeasure $\widetilde{C}$ access $P$ as a black-box and computes the function $f$ using the random self-reducible properties of $f$.
    }%
  }
}
\vskip2pt

\section{Implementation of Countermeasures}\label{sec:implementation-of-countermeasures}

Table~\ref{tab:finite-field-examples} lists all functions of some finite field operations and their corresponding random self-reducible properties. In this section, we examplify each function that we used in end-to-end experiments and show how to apply the countermeasure approach to protect the function. 

\begin{table*}[t]
  \small
  \centering
  % \footnotesize
  \setlength\cellspacetoplimit{2pt}
  \setlength\cellspacebottomlimit{2pt}
  \caption{Random Self-Reducible Properties. $x$ and $y$ are integers and $p$ and $q$ are polynomial.
  %  $x \gets \tilde{x}_1 + \tilde{x}_2$, $y \gets \tilde{y}_1 + \tilde{y}_2$, $p \gets \tilde{p}_1 + \tilde{p}_2$, and $q \gets \tilde{q}_1 + \tilde{q}_2$.
   } \label{tab:finite-field-examples}
  \begin{NiceTabular}{|Ol|Ol|Ol|}
    \toprule
    {\bf Program $P$}                       & {\bf Function $f$}             & {\bf Random Self-Reducible Property}                                                                                                                        \\\hline
    \hline Mod Operation                    & $f(x, R)=x \bmod R$            & $P(x, R) \gets P(\tilde{x}_1, R) +_R P(\tilde{x}_2, R)$                                                                                                     \\
    \hline Modular Multiplication           & $f(x, y, R)=x \cdot_R y$       & $P(x, y, R) \gets  P(\tilde{x}_1, \tilde{y}_1, R) +_R P(\tilde{x}_2, \tilde{y}_1, R) +_R P(\tilde{x}_1, \tilde{y}_2, R) +_R P(\tilde{x}_2, \tilde{y}_2, R)$ \\
    \hline Modular Exponentiation           & $f(a, x, R)=a^x \bmod R$       & $P(a, x, R) \gets P(a, \tilde{x}_1, R) \cdot_R P(a, \tilde{x}_2, R)$                                                                                        \\
    \hline Modular Inverse                  & $f(x, R) \cdot_R x=1$          & $P(x, R) \gets \tilde{w} \cdot_R P(x \cdot_R \tilde{w})$ where $P(\tilde{w}, R) \cdot_R \tilde{w} = 1$ and $P(x, R) \cdot_R x = 1$                          \\
    \hline
    \hline Polynomial Multiplication        & $f(p_x, q_x)=$ $p_x \cdot q_x$ & $P(p, q) \gets P(\tilde{p}_1, \tilde{q}_1)+P(\tilde{p}_2, \tilde{q}_1)+P(\tilde{p}_1, \tilde{q}_2)+P(\tilde{p}_2, \tilde{q}_2)$                             \\
    \hline Number Theoretic Transform  & $f(x_1, \ldots, x_n) = \cdots$ & $P(x_1, \ldots, x_n) \gets P(x_1+\tilde{r}_1, \ldots, x_n+\tilde{r}_n)-P(\tilde{r}_1, \ldots, \tilde{r}_n)$                                                 \\
    \hline
    \hline Integer Multiplication           & $f(x, y)=x \cdot y$            & $P(x, y) \gets  P(\tilde{x}_1, \tilde{y}_1) + P(\tilde{x}_1, \tilde{y}_2) + P(\tilde{x}_2, \tilde{y}_1) + P(\tilde{x}_2, \tilde{y_2})$                      \\
    \hline Integer Multiplication           & $f(x, y)=x \cdot y$            & $P(x, y) \gets P(x+\tilde{r}, y+\tilde{s})-P(\tilde{r}, y+\tilde{s})-P(x+\tilde{t}, \tilde{s})+P(\tilde{t}, \tilde{s})$                                     \\
    \hline Integer Division                 & $f(x, R)=x \div R$             & $\displaystyle P(x, y) \gets  P(x_1, R)+P(x_2, R) + P\left(P_{\text{mod}}(x_1, R)+P_{\text {mod }}(x_2, R), R\right)$                                       \\
    \hline
    \hline Matrix Multiplication            & $f(A, B) = A \times B$         & $P(A, B) \leftarrow P(\tilde{A}_1, \tilde{B}_1)+P(\tilde{A}_2, \tilde{B}_1)+P(\tilde{A}_1, \tilde{B}_2)+P(\tilde{A}_2, \tilde{B}_2)$                        \\
    \hline Matrix Inverse                   & $f(A) = A^{-1}$                & $P(A) \gets \tilde{R} \times P(A \times \tilde{R})$ where $A$ and $\tilde{R}$ are invertible $n$-by-$n$ matrices.                                           \\
    \hline Matrix Determinant               & $f(A) = \det A$                & $P(A) \gets P(\tilde{R}) / P(A \times \tilde{R})$ where $\tilde{R}$ is invertible.                                                                          \\
    \bottomrule
  \end{NiceTabular}
\end{table*}

\paragraph*{Random Split Function.}\label{sec:random-split-function}
A random splitter is used to divide the input \( x \) into \( c \) shares \( a_1, \ldots, a_c \) such that \( x = a_1 + \cdots + a_c \). Algorithm~\ref{alg:random-split} provides a possible implementation of the random splitter, accommodating an additional input that specifies the total number of shares.

\begin{algorithm}\label{alg:random-split}
  \caption{Random-Split$(m, x, c)$.}
  \LinesNumbered
  \small

  \KwIn{modulus: $m$, input value: $x$, \# of shares: $c$}
  \KwOut{an array of shares $a_1, a_2, \ldots, a_c$.}

  \BlankLine
  Initialize an array $s[1 \ldots c]$ and initialize $sum \gets 0$\;
  $i \gets 1$\;
  \For{$i$ \KwTo $c - 1$}{
    $s[i] \getsdollar \text{random integer in } \mathbb{Z}_m$\;
    $sum \gets sum + s[i]$
  }
  % \lIf{$i \ne c-1$}{output "FAIL" and halt \Comment*[f]{check loop invariant}}
  $s[c] \gets x - sum \pmod{m}$\;
  \Return{$s$}
\end{algorithm}

This algorithm ensures that the sum of all shares \( a_1, a_2, \ldots, a_c \) is congruent to the original input \( x \) modulo \( m \). This congruence condition is vital for maintaining the integrity of the split and ensuring that the original input can be accurately reconstructed from the shares.

\paragraph*{Majority Vote Function.}\label{sec:majority-function}
The function {majority()} selects the value that occurs most frequently; in case of a tie it selects the first such value. In practice, if the majority() function does not get all the same values, then it would at least ``log'' that some error has been detected. Here,  it is unlikely that the majority of the values are wrong.

\begin{algorithm} \label{alg:majority-element-1}
  \caption{Majority Vote Algorithm}
  \small
  \SetKwInOut{Input}{Input}
  \SetKwInOut{Output}{Output}
  \SetKw{Throw}{throw}
  \LinesNumbered
  \Input{A list of elements $a_1, a_2, \dots, a_n$}
  \Output{The majority element of the list}

  \BlankLine
  Initialize an element $m$ and a counter $i$ with $i = 0$;

  \BlankLine
  \For{$j \gets 1$ to $n$}{
    \lIf{$i = 0$}{
      $m \gets a_j$ and $i \gets 1$
    }\lElseIf{$m = a_j$}{
      $i \gets i + 1$
    }\lElse{
      $i \gets i - 1$
    }
  }
  % \BlankLine
  \Return $m$
\end{algorithm}

The software implementation of all ($c, n$)-secure programs use Boyer-Moore majority\footnote{\url{https://www.cs.utexas.edu/~moore/best-ideas/mjrty}} vote algorithm~\cite{boyer1991mjrty} as the majority function implementation. Algorithm~\ref{alg:majority-element-1} maintains in its local variables a sequence element $m$ and a counter $i$, with the counter initially zero. It then processes the elements of the sequence, one at a time. When processing an element $a_j$, if the counter is zero, the algorithm stores $a_j$ as its remembered sequence element and sets the counter to one. Otherwise, it compares $a_j$ to the stored element and either increments the counter (if they are equal) or decrements the counter (otherwise). At the end of this process, if the sequence has a majority, it will be the element stored by the algorithm.
Algorithm~\ref{alg:majority-voting-self-correcting} is a self-correcting version of the majority vote algorithm. It calls the majority function $n$ times and at each iteration, it shuffles the input list to ensure that the input list is random at each iteration against simple power side-channel attacks. The algorithm returns the majority element of the list, if it exists. 

\begin{algorithm}
  \caption{Protected Majority Vote ($\ell, \operatorname{majority}, n$)}
  \label{alg:majority-voting-self-correcting}
  \small
  \SetKwInOut{Input}{Input}
  \SetKwInOut{Output}{Output}
  \Input{Votes $\ell = a_1, a_2, \dots, a_n$, function $\operatorname{majority}$, and $n$}
  \Output{The majority element of the list, if it exists}
  \BlankLine
  % $N \gets 12 \ln (1 / \beta)$ \\
  $m \gets 1$ \\
  \For{$m$ \KwTo $n$}{
    $\ell_1 \getsdollar \operatorname{shuffle}(\ell)$ \\
    $\text{answer}_m \gets \operatorname{majority}(\ell_1)$\\
  }
  \lIf{$m \ne n$}{output "FAIL" and halt \Comment*[f]{verify loop completion}}
  % \BlankLine
  \Return the majority in $\{\text{answer}_m\text{: } m=1, \ldots, n\}$
\end{algorithm}

The Fisher-Yates shuffle algorithm is used to shuffle the input list at each iteration~\cite{fisher1963statistical, durstenfeld1964algorithm}. It iterates through a sequence from the end to the beginning (or the other way) and for each location $i$, it swaps the value at $i$ with the value at a random target location $j$ at or before $i$.

\begin{algorithm}
  \caption{Fisher-Yates Shuffle}
  \small
  \SetKwInOut{Input}{Input}
  \SetKwInOut{Output}{Output}
  \Input{A list of elements $a_1, a_2, \dots, a_n$}
  \Output{A random permutation of the elements in the input list}
  \label{alg:fisher-yates-shuffle}

  \BlankLine
  \For{$i \gets n-1$ down to $1$}{
    Choose a random integer $j$ such that $0 \leq j \leq i$\;
    Swap $a_i$ and $a_j$\;
  }

  % \BlankLine
  \Return the shuffled list \;
\end{algorithm}

Line 5 of Algorithm~\ref{alg:majority-voting-self-correcting} verifies that the loop has completed. This is a simple check to ensure that the loop has completed and that the algorithm has not been interrupted. This is a classical countermeasure against instruction skip type of fault injection attacks~\cite{timmers2016controlling}. In the rest of the countermeasures, we don't need to check because the self-correctness property is already guaranteed by the majority function. 

\paragraph*{Mod Function Countermeasure}\label{sec:mod-function-1}
We consider computing an integer $\bmod\,R$ for a positive number $R$. In this case, $f(x, R)=x \bmod R$. Algorithm~\ref{alg:mod-function} shows a pseudocode for the protected mod function with a security parameter of 2.

\begin{algorithm}\label{alg:mod-function}
  \caption{2-secure protected mod operation $(P, R, x)$}
  \small
  % \SetKwInOut{Input}{Ensure}
  % \Input{$x=x_1+x_2- cR2^n$  and  $x \bmod R = (x_1 \bmod R) +_R (x_2 \bmod R)$}

  % \BlankLine
  % $N \gets 12 \ln (1 / \beta)$ \\
  % \For{$m=1, \ldots, N$}{
  $x_1, x_2$ $\getsdollar $ Random-Split($R 2^n, x$)\;
  % $\text{answer}_m \gets P(x_1, R) +_R P(x_2, R)$
  % $\text{answer} \gets P(x_1, R) +_R P(x_2, R)$\;
  \Return $P(x_1, R) +_R P(x_2, R)$\;
  % }
  % \Return the majority in $\{\text{answer}_m\text{: } m=1, \ldots, N\}$
  % \Return answer
\end{algorithm}

However, Algorithm~\ref{alg:mod-function-3} shows a version for a security parameter 3 in which we increase the security parameter by one by increasing the number of shares to three. All random self-reducible properties in Table~\ref{tab:finite-field-examples} are applicable to increase shares to $n$.

\begin{algorithm}\label{alg:mod-function-3}
  \caption{3-secure protected mod operation $(P, R, x)$.}
  \small
  $x_1, x_2, x_3$ $\getsdollar $ Random-Split($R 2^n, x$)\;
  \Return $P(x_1, R) +_R P(x_2, R) +_R P(x_3, R)$\;
\end{algorithm}

\paragraph*{Modular Multiplication Countermeasure}\label{sec:modular-multiplication-1}

We now consider multiplication of integers $\bmod\ R$ for a positive number $R$. In this case, $f(x, y, R)=x \cdot_R y$. Suppose that both $x$ and $y$ are in the range $\mathbb{Z}_{R 2^n}$ for some positive integer $n$. Algorithm~\ref{alg:mod-mult} shows a possible implementation for the protected modular multiplication with a $c$ security parameter set to 2.

\begin{algorithm}\label{alg:mod-mult}
  \caption{2-secure mod. multiplication $(P, R, x, y)$}
  \small
  $x_1, x_2$ $\getsdollar$  Random-Split($R \times 2^n, x$) \;
  $y_1, y_2$ $\getsdollar$  Random-Split($R \times 2^n, y$) \;
  \Return $P(x_1, y_1, R) +_R P(x_2, y_1, R) +_R P(x_1, y_2, R) + P(x_2, y_2, R)$
\end{algorithm}

\paragraph*{Modular Exponentiation Countermeasure}\label{sec:modular-exponentiation-1}
We now consider exponentiation of integers $\bmod R$ for a positive number $R$. In this case, $f(a, x, R)=a^x \bmod R$. We restrict attention to the case when $\operatorname{gcd}(a, R)=1$ and when we know the factorization of $R$, and thus we can easily compute $\phi(R)$, where $\phi$ is Euler's function. Suppose that $x$ is in the range $\mathbb{Z}_{\phi(R) 2^n}$.

\begin{algorithm}\label{alg:mod-exp}
  \small
  % \caption{Protected Modular Exponentiation $(R, a, x, \beta)$.}
  \caption{2-secure mod. exponentiation $(P, R, a, x)$}
  % $N \gets 12 \ln (1 / \beta)$ \\
  % \For{$m=1, \ldots, N$}{
  $x_1, x_2$ $\getsdollar$ Random-Split($\phi(R)2^n, x$) \;
  % $\text{answer}_m \gets P(a, x_1, R) \cdot_R P(a, x_2, R)$ \tcp*[l]{calls modular multiplication}
  \Return $\gets P(a, x_1, R) \cdot_R P(a, x_2, R)$  \Comment*[r]{calls Algo.~\ref{alg:mod-mult}}
  %\tcp*[l]{calls modular multiplication}\;
  % }
  % \Return the majority in $\{\text{answer}_m\text{: } m=1, \ldots, N\}$
  % \Return answer 
\end{algorithm}

The modular exponentiation self-correcting program is very simple to code. The hardest operation to perform is the modular multiplication $P\left(a, x_1, R\right) \cdot_R$ $P\left(a, x_2, R\right)$.
The self-correcting program can compute this multiplication directly without using random self-reducible property, however, for extra protection, $2$-secure modular multiplication can be used (cf. Algorithm~\ref{alg:mod-mult}).

\paragraph*{Polynomial Multiplication Countermeasure}\label{sec:poly-mult}
We consider the multiplication of polynomials over a ring. Let $R^d[x]$ denote the set of polynomials of degree $d$ with coefficients from some ring $R$, and let $\mathbb{U}_{R^d[x] \times R^d[x]}$ be the uniform distribution on $R^d[x] \times R^d[x]$. In this case, $f(p(x), q(x))=$ $p(x) \cdot q(x)$, where $p, q \in R^d[x]$.

\begin{algorithm}
  \caption{2-secure polynomial multiplication $(P, p, q)$}\label{alg:poly-mult}
  \small
  % \SetKwInOut{Input}{Input}
  % \SetKwInOut{Output}{Output}
  % \Input{$(p, q, \beta)$}
  % \Output{Decrypted Message: $m$}
  % $N \gets 12 \ln (1 / \beta)$\\
  % \For{$m=1, \ldots, N$}{
  Choose $p_1 \in_{\mathbb{U}} R^d[x]$ \Comment*[r]{random polynomial}
  Choose $q_1 \in_{\mathbb{U}} R^d[x]$ \Comment*[r]{random polynomial}
  $p_2 \gets p-p_1$\;
  $q_2 \gets q-q_1$\;
  \Return $P\left(p_1, q_1\right)+P\left(p_2, q_1\right)+P\left(p_1, q_2\right)+P\left(p_2, q_2\right)$\\
  % }
  % \Return the majority in $\{\text{answer}_m\text{: } m=1, \ldots, N\}$
\end{algorithm}

\paragraph*{Number Theoretic Transforms}\label{sec:ntt}

Transforms used in signal processing such as the Fast Fourier Transform (FFT) or Number Theoritic Transform (NTT) or their inverse can be protected with our countermeasure. NTT over an $n$ point sequence is performed using the well-known butterfly network, which operates over several layers/stages. The atomic operation within the NTT computation is denoted as the butterfly operation. A butterfly operation takes as inputs $(a, b) \in \mathbb{Z}_q^2$ and a twiddle constant $w$, and produces outputs $(c, d) \in \mathbb{Z}_q^2$.
An NTT/INTT of size $n=2^k$ typically consists of $k$ stages with each stage containing $n / 2$ butterfly operations. Figure~\ref{fig:ntt-butterfly} shows the data-flow graph of a butterfly-based NTT for an input sequence with length $n=8$. All operations are linear in nature, and thus, the NTT/INTT can be viewed as a linear function.

\begin{figure}[h]
  \centering
  \includegraphics[width=.8\columnwidth]{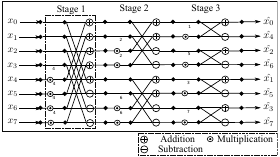}
  \caption{Data flow graphs of a butterfly-based NTT for size $n=8$~\cite{ravi2023fiddling}.}
  \label{fig:ntt-butterfly}
\end{figure}

\begin{lemma}\label{lemma:groups}
  Let $G$ be an abstract finite group under the operation $\circ$, and let $x$ be an arbitrary value from the group. If $\tilde{r}$ is a uniform random value, then so is $x \circ \tilde{r}$.
\end{lemma}

\begin{proof}
  Consider the function $f(z)=x \circ z$. Since $G$ is a group this is a one-to-one onto function. Thus, if $\tilde{r}$ is selected randomly, then so is $f(\tilde{r})$.
\end{proof}

Consider such a transformation $T\left(x_1, \ldots, x_n\right)$ where the values $x_i$ are fixed point numbers, i.e., 2-complement's arithmetic of some fixed size. This follows since the transformation is linear. Thus,
$
  T\left(x_1, \ldots, x_n\right)=T\left(x_1+\tilde{r}_1, \ldots, x_n+\tilde{r}_n\right)-T\left(\tilde{r}_1, \ldots, \tilde{r}_n\right)
$.

\begin{algorithm}
  \caption{2-secure NTT ($P, x_1, \ldots, x_n \in \mathbb{Z}_q^2$).}\label{alg:ntt}
  \small
  Choose $\tilde{r}_1, \ldots, \tilde{r}_n \in_{\mathbb{U}} \mathbb{Z}_q^2$\\
  \Return $\text{NTT}\left(x_1 +\tilde{r} _1, \ldots, x_n +\tilde{r} _n\right) - \text{NTT}\left(\tilde{r} _1, \ldots, \tilde{r} _n\right)$\;
\end{algorithm}

The key point here is that since fixed-point values are a group under addition, the value $x_i+\tilde{r}_i$ is a uniform random value by Lemma~\ref{lemma:groups}. Note that the function basis consists of just addition~\cite{lipton1991new}.
The countermeasure for NTT is given in Algorithm~\ref{alg:ntt}.

\section{End-to-End Implementations}\label{sec:implementation}
In this section, we introduce implementations of the RSA-CRT signature algorithm and Kyber's key generation algorithm, detailing existing vulnerabilities and how we can protect them against them using our methods.

\paragraph*{Securing RSA-CRT Algorithm.}

RSA is a cryptographic algorithm commonly used in digital signatures and SSL certificates. Due to the security of RSA, which relies on the difficulty of factoring the product of two large prime numbers, the calculation of RSA is relatively slow. Therefore, it is seldom used to encrypt the data directly.

For efficiency, many popular cryptographic libraries (e.g., OpenSSL) use RSA based on the Chinese remainder theorem(CRT) for encryption or signing messages. Algorithm \ref{alg:RSA-CRT} is the RSA-CRT signature generation algorithm. With the private key, we pre-calculate the values $d_p = d \mod (p-1)$, $d_q = d \mod (q-1)$ and $u = q^{-1} \mod p $, then generate the intermediate value $s_p = m^{d_p} \mod p$, $s_q = m^{d_q} \mod q$. Finally, combine two intermediate value $s_p$, $s_q$ with the Garner's algorithm $S = s_q + (((s_p-s_q)\cdot u) \mod p) \cdot q$
The RSA based on CRT is about four times faster then classical RSA.

\begin{algorithm}%[H]
  \small
  \DontPrintSemicolon
  \LinesNumbered  % Enable line numbers\
  \KwIn{A message $M$ to sign, the private key $(p,q,d)$, with $p > q$, pre-calculated values $d_p = d \mod (p - 1)$, $d_q = d \mod (q - 1)$, and $u = q^{-1} \mod p$.}
  \KwOut{A valid signature $S$ for the message $M$.}

  \BlankLine
  $m \gets$ Encode the message $M$ in $m \in \mathbb{Z}_N$\; %(with PKCS\#1)\;
  $s_p \gets m^{d_p} \mod p$\Comment*[r]{\red{Protection with Algorithm~\ref{alg:mod-exp}}}
  $s_q \gets m^{d_q} \mod q$\Comment*[r]{\red{Protection with Algorithm~\ref{alg:mod-exp}}}
  $t \gets s_p - s_q$\;
  \If{$t < 0$}{
    $t \gets t + p$\;
  }
  $S \gets s_q + ((t \cdot u) \mod p) \cdot q$\;
  \Return{$S$ as a signature for the message $M$}\;

  \caption{RSA-CRT Signature Generation Algorithm}
  \label{alg:RSA-CRT}
\end{algorithm}

However, using CRT to improve RSA operation efficiency makes RSA vulnerable. For instance, in \cite{aumuller2003fault}, Aum\"{u}ller et al. provided the fault-based cryptanalysis method of RSA-CRT that the attacker can intentionally induce the fault during the computation, which changes $s_p$ to faulty $\hat{s_p}$, to obtain the faulty output and factorize $N$ by using the equation {\small $q = gcd((s'^e-m) \mod N, N)$} to recover the secret key. Sung-Ming et al. provided another equation that can factorize $N$ with faulty signature in \cite{yen2003hardware}. There are two scenarios that the attacker can break the RSA-CRT. If the attacker knows the value of the message and faulty output, they can factorize $N$ with the previous equation. On the other hand, if the attacker knows the value of correct and faulty signatures, they can factorize $N$ with the equation {\small $q = gcd((\hat{s}-s) \mod N, N)$}.

We protect Line 2 and Line 3 of Algorithm \ref{alg:RSA-CRT} using the proposed countermeasure against the attack introduced in~\cite{aumuller2003fault}.

\paragraph*{Securing Kyber Key Generation Algorithm.}

The NIST standardization process for post-quantum cryptography~\cite{nist:pqc} has finished its third round, and provided a list of new public key schemes for new standardization~\cite{alagic2022status}. While implementation performance and theoretical security guarantees served as the main criteria in the initial rounds, resistance against side-channel attacks (SCA) and fault injection attacks (FIA) emerged as an important criterion in the final round, as also clearly stated by NIST at several instances~\cite{ravi2021sidechannel}.

\begin{algorithm}%[H]
  \DontPrintSemicolon
  \small
  \LinesNumbered  % Enable line numbers
  % \KwOut{output}

  \BlankLine
  $seed_A \gets$ Sample$_U()$\; %\tcp*{Generate uniform $Seed_A$}
  $seed_B \gets$ Sample$_U()$\; %\tcp*{Generate uniform $Seed_B$}
  $\hat{A} \gets$ NTT($A$)\; %\tcp*{Expand $seed_A$ into $\hat{A}$ in NTT domain}
  $s \gets$ Sample$_B(seed_B, coins_s)$\; %\tcp*{Sample secret $s$ using $(Seed_B, coins_s)$}
  $e \gets$ Sample$_B(seed_B, coins_e)$\; %\tcp*{Sample error $e$ using $(Seed_B, coins_e)$}
  $\hat{s} \gets$ NTT($s$)\Comment*[r]{\red{Protection with Algorithm~\ref{alg:ntt}}} %\tcp*{NTT($s$)}
  $\hat{e} \gets$ NTT($e$)\;%\Comment*[r]{\red{Protect with Algorithm~\ref{alg:ntt}}} %\tcp*{NTT($e$)}
  $\hat{t} \gets \hat{A} \odot \hat{s} + \hat{e}$\; %\tcp*{$t = A \cdot s + e$ in NTT domain}
  \Return{$pk = (seed_A, \hat{t}), sk = (\hat{s})$}\;

  \caption{CPA Secure Kyber PKE (CPA.KeyGen)}
  \label{alg:kyber-pke}
\end{algorithm}

They typically operate over polynomials in polynomial rings, and notably, polynomial multiplication is one of the most computationally intensive operations in practical implementations of these schemes. Among the several known techniques for polynomial multiplication such as the schoolbook multiplier, Toom-Cook~\cite{toom1963complexity} and Karatsuba~\cite{karatsuba1963multiplication}, the Number Theoretic Transform (NTT) based polynomial multiplication~\cite{cooley1965algorithm} is one of the most widely adopted techniques, owing to its superior run-time complexity. Over the years, there has been a sustained effort by the cryptographic community to improve the performance of NTT for lattice-based schemes on a wide-range of hardware and software platforms~\cite{roy2014compact, poppelmann2015high, botros2019memory, abdulrahman2021multi, chung2021ntt}. As a result, the use of NTT for polynomial multiplication yields the fastest implementation for several lattice-based schemes. In particular, the NTT serves as a critical computational kernel used in Kyber~\cite{avanzi2020crystalskyber} and Dilithium~\cite{lyubashevsky2017crystalsdilithium}, which were selected as the first candidates for PQC standardization~\cite{ravi2023fiddling}.

\begin{figure}[ht]
  \centering
  \includegraphics[width=\columnwidth]{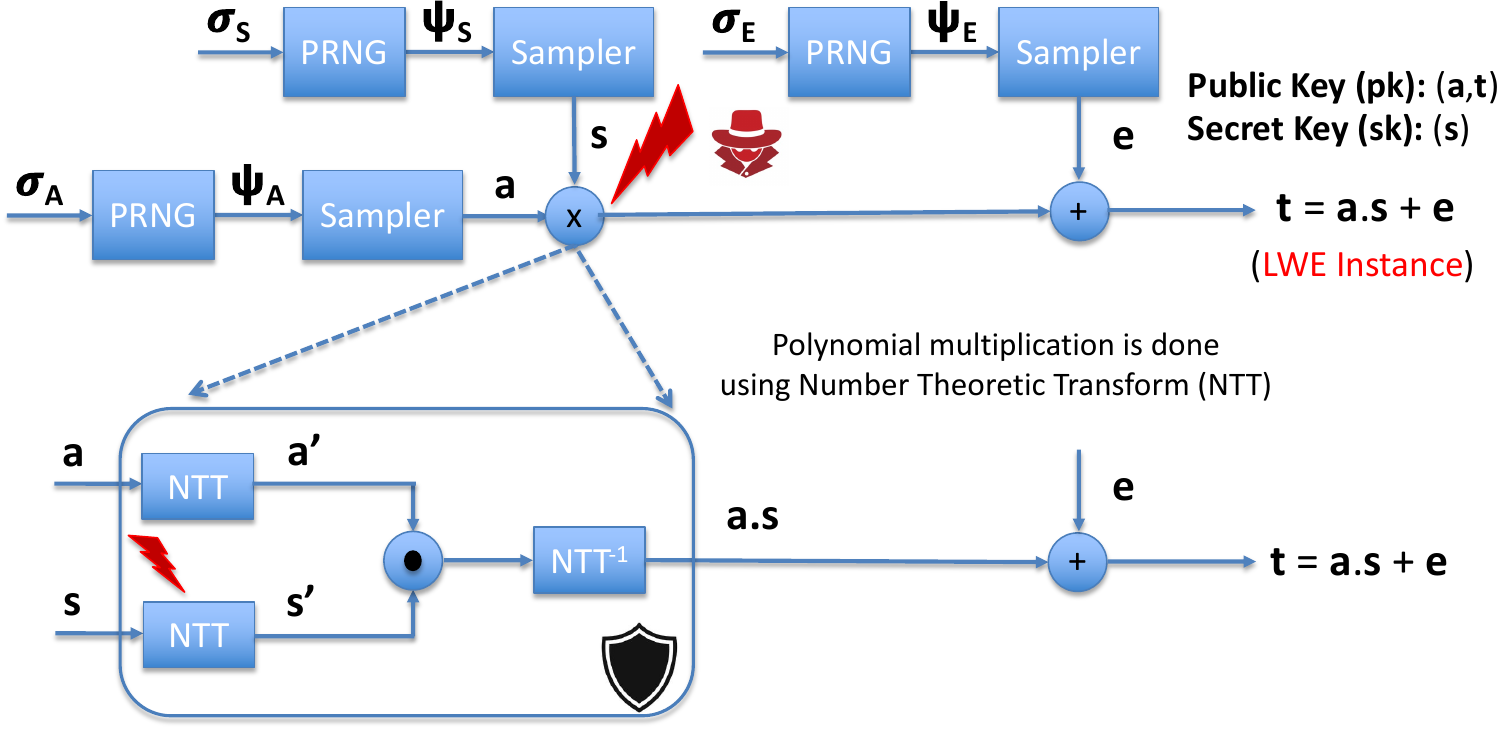}
  \caption{In Kyber Key Generation algorithm polynomial multiplication is done using Number Theoretic Transform (NTT). The NTT is protected using the proposed countermeasure against the attack introduced in~\cite{ravi2023fiddling}.}
  \label{fig:kyber-poly-ntt-sketch}
\end{figure}

Figure~\ref{fig:kyber-poly-ntt-sketch} illustrates a recent fault injection attack~\cite{ravi2023fiddling} that exposes a significant vulnerability in NTT-based polynomial multiplication, allowing the zeroization of all twiddle constants through a single targeted fault. This vulnerability enables practical key/message recovery attacks on Kyber KEM and forgery attacks on Dilithium. Moreover, the proposed attacks are also shown to bypass most known fault countermeasures for lattice-based KEMs and signature schemes. 

To safeguard polynomial multiplication, we can employ Algorithm~\ref{alg:poly-mult} or protect individual NTT operations using Algorithm~\ref{alg:ntt}. In this paper, we focus on securing the NTT operation targeted by Ravi et al.\cite{ravi2023fiddling} using Algorithm\ref{alg:ntt}. Consequently, we reinforce Line 6 of Algorithm~\ref{alg:kyber-pke} with our proposed countermeasure against the attack delineated in~\cite{ravi2023fiddling}.

\section{Evaluation}\label{sec:experiments}
We conducted three experimental sets to assess our countermeasure's effectiveness against fault injection and power side-channel attacks. Initially, we evaluated protected operations individually, including modular multiplication, modular exponentiation, and NTT. Subsequently, we assessed our countermeasure's robustness within RSA-CRT and Kyber key generation algorithms. Finally, we examined the latency overhead introduced by our countermeasure.

To capture power traces, for our experiments we use an ATSAM4S-based target board. SAM4S is a microcontroller based around the 32-bit ARM cortex-m4 processor core, which is commonly used in embedded systems such as IoT devices. The specific target board comes with the ChipWhisperer Husky~\cite{2014:204}, which is the equipment that we used for power trace collection.

The voltage fault injection test bed is created using Riscure's VC Glitcher product\footnote{\url{https://www.riscure.com/products/vc-glitcher/}} that generates an arbitrary voltage signal with a pulse resolution of 2 nanoseconds. We use a General Purpose Input Output (GPIO) signal to time the attack which allows us to inject a glitch at the moment the target is executing the targeted code. The target's reset signal is used to reset the target prior to each experiment to avoid data cross-contamination. All fault injection experiments are performed targeting an off-the-shelf development platform built around an {STM32F407 MCU}, which includes an ARM Cortex-M4 core running at {168 MHz}. This Cortex-M4 based MCU has an instruction cache, a data cache and a prefetch buffer.

\paragraph*{Power Side-Channel Attack Evaluation.}
In power side-channel evaluation, we use the Hamming Weight leakage model and the Test Vector Leakage Assessment (TVLA)~\cite{tvla} to evaluate the effectiveness of our countermeasure.
The instantaneous power consumption measurement corresponding to a single execution of the target algorithm is referred to as power trace. Each power trace is therefore a vector of power samples, and the t-test has to be applied sample-wise. The obtained vector is referred to as t-trace.

To detect Points-of-Interest,
we employ the Sum of Squared pairwise T-differences (SOST)\cite{template} method, setting the threshold at 20\% of the maximum. The t-test window size is uniformly set to $\pm 8$ for all operations. We define the power side-channel security parameter as $c=2$ in the $c$-secure countermeasure in Algorithm~\ref{alg:psca-countermeasure} applicable to all operations. In the mod operation and modular multiplication, the entire operation is targeted, while in modular exponentiation and NTT, attacks are focused on the constant-time Montgomery ladder~\cite{montgomery1985modular, liu2010efficient} modular exponentiation function. For TVLA analysis, two sets of test vectors were created: one with random numbers of Hamming weight 12 and another with a Hamming weight of 4, using 1000 random numbers for each. These vectors were used for evaluating both protected and unprotected cryptographic operations.

In our study, we also evaluated the distinguishability of total power consumption in modular operations and modular multiplication. For modular multiplication, we maintained one operand's value constant while varying the other operand among numbers with different Hamming Weights. This approach enables a comparative analysis of power consumption patterns in modular operations, particularly between unprotected and protected versions, offering insights into how variations in Hamming Weight influence power consumption in these protected cryptographic operations.

Our evaluation indicates that the RSR countermeasure significantly reduced t-test results, bringing them into acceptable regions. For example, in the mod operation, the maximum t-test result decreased from 415.7 to 4.12, and for NTT, it dropped from 417.7 to 7.69. These results, which are detailed in Table~\ref{fig:psca-ttest}, demonstrate an average reduction of two orders of magnitude, highlighting the effectiveness of the RSR countermeasure in enhancing the security of cryptographic operations against side-channel attacks.

\begin{figure*}
    \centering
    \subfigure[Unprotected Mod Operation]{
        \includegraphics[width=0.23\textwidth]{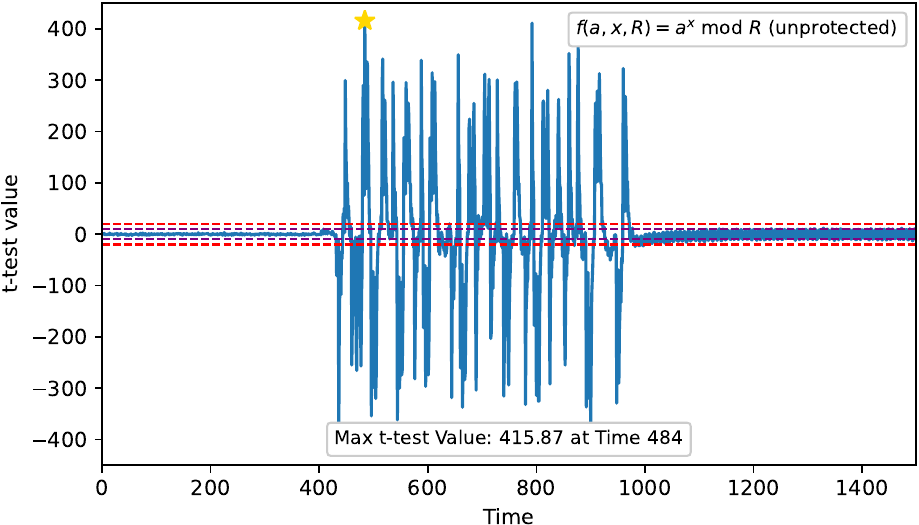}
        \label{fig:t-tests:poly-mult-unprotected}
    }
    \subfigure[\textcolor{darkgreen}{Protected} Mod Operation]{
        \includegraphics[width=0.23\textwidth]{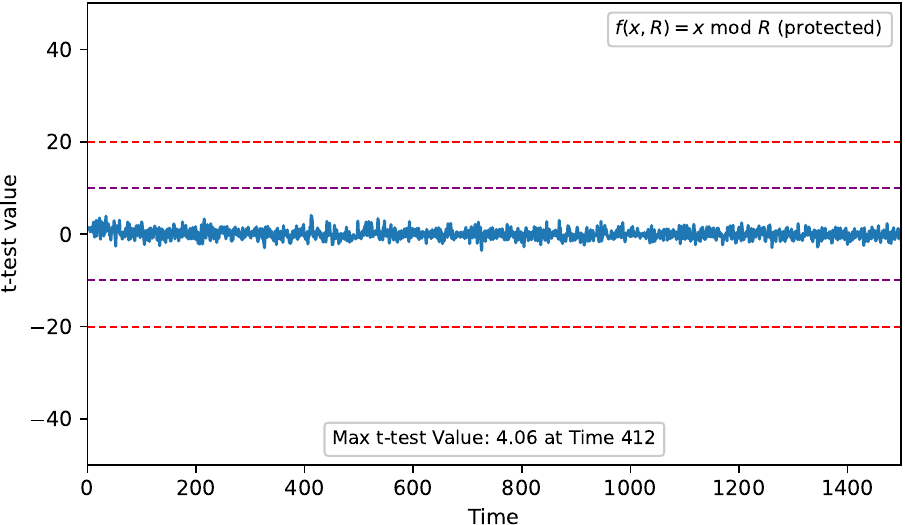}
        \label{fig:t-tests:poly-mult-protected}
    }
    \subfigure[Unprotected Mod. Mult.]{
        \includegraphics[width=0.23\textwidth]{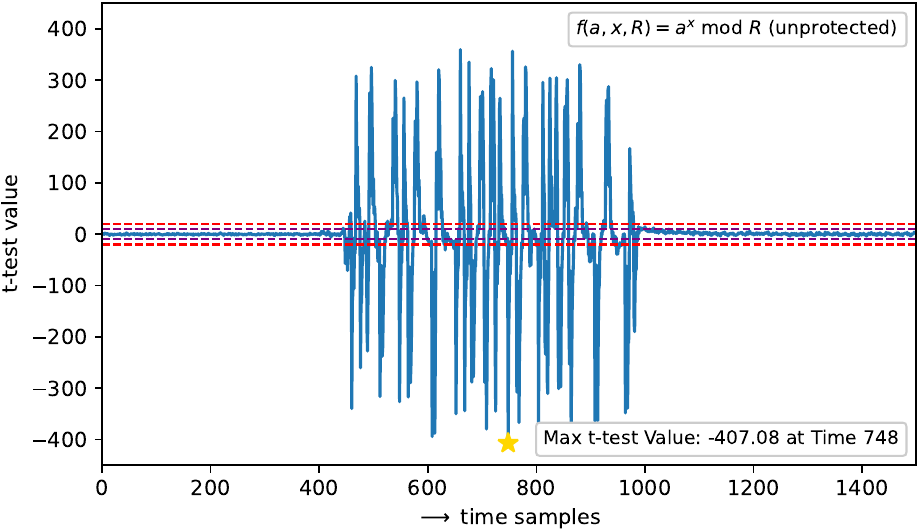}
        \label{fig:ttest:mod-mult-unprotected}
    }
    \subfigure[\textcolor{darkgreen}{Protected} Mod. Mult.]{
        \includegraphics[width=0.23\textwidth]{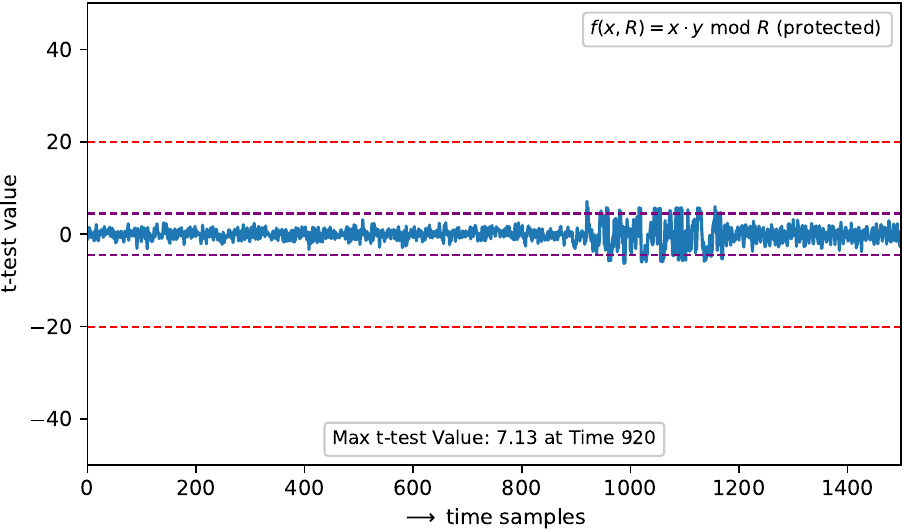}
        \label{fig:ttest:mod-mult-protected}
    }
    \subfigure[Unprotected Mod. Exp.]{
        \includegraphics[width=0.23\textwidth]{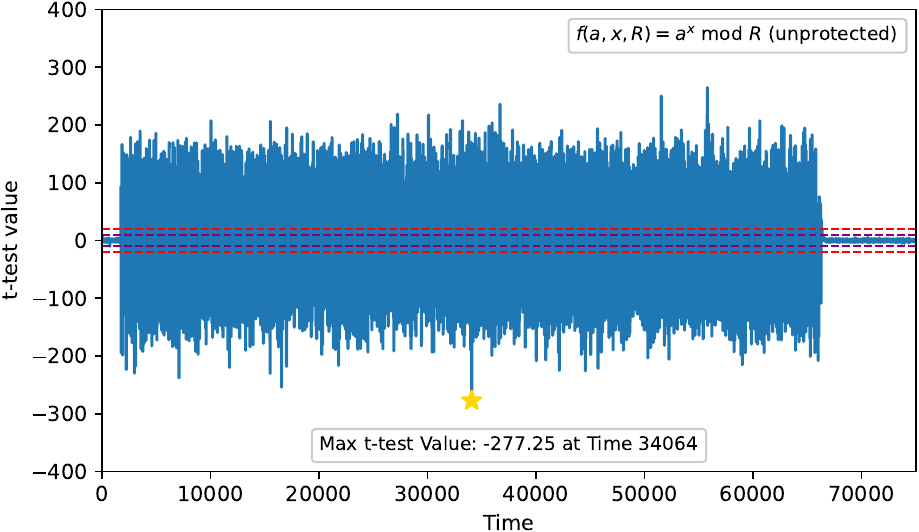}
        \label{fig:ttest:mod-exp-unprotected}
    }
    \subfigure[\textcolor{darkgreen}{Protected} Mod. Exp.]{
        \includegraphics[width=0.23\textwidth]{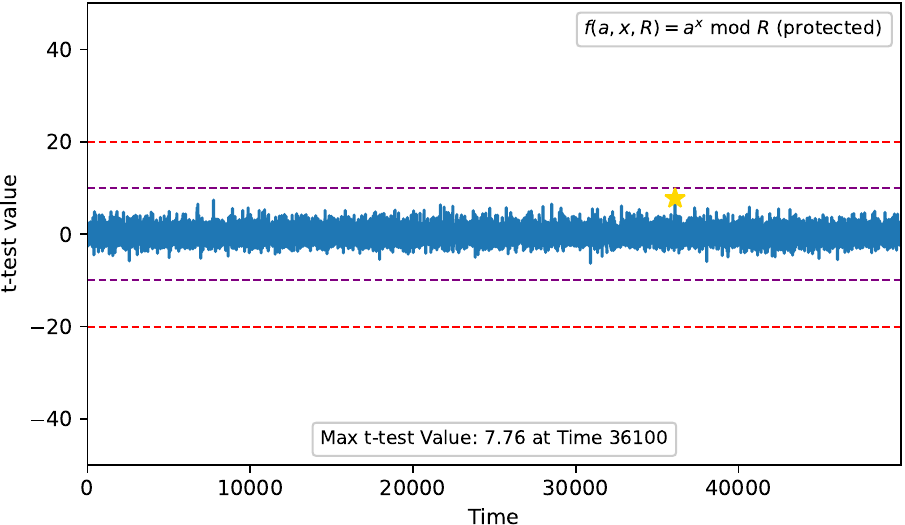}
        \label{fig:ttest:mod-exp-protected}
    }
    \subfigure[Unprotected NTT]{
        \includegraphics[width=0.23\textwidth]{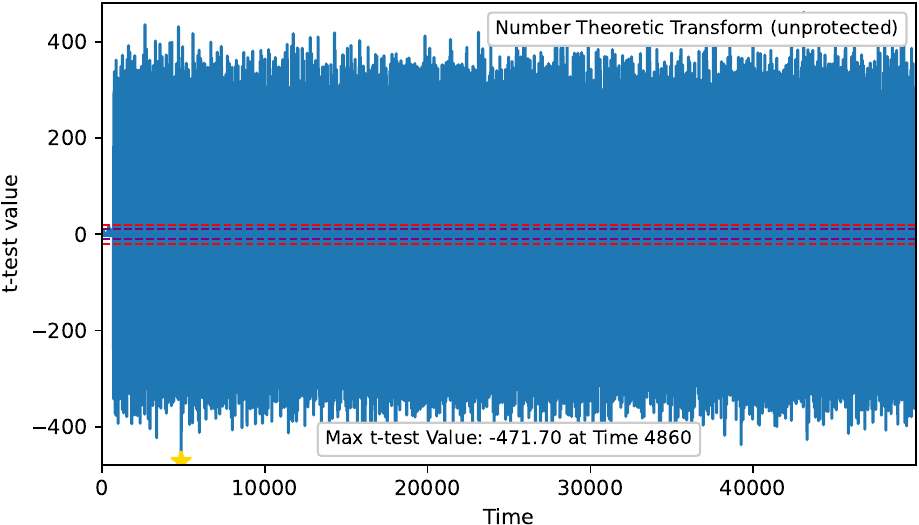}
        \label{fig:t-tests:ntt-unprotected}
    }
    \subfigure[\textcolor{darkgreen}{Protected} NTT]{
        \includegraphics[width=0.23\textwidth]{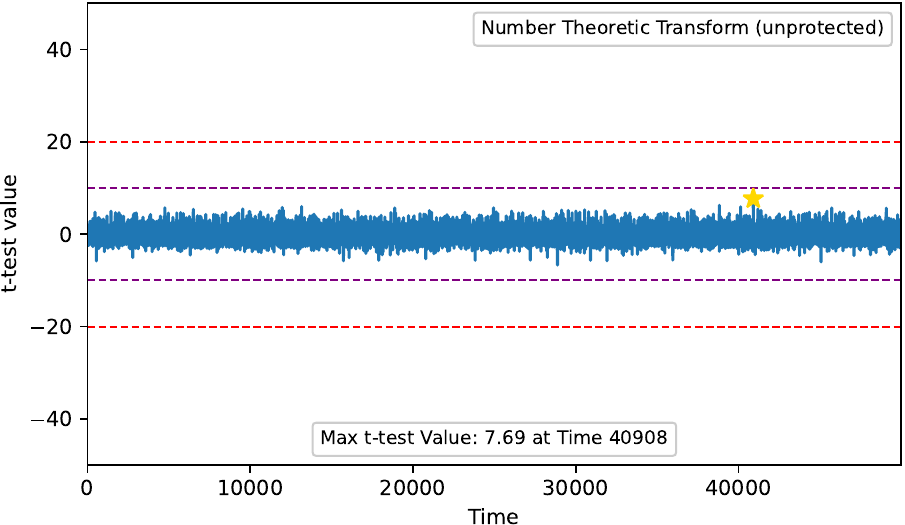}
        \label{fig:t-tests:ntt-protected}
    }
    \caption{Power Side-Channel Attack Evaluation t-tests}
    \label{fig:psca-ttest}
\end{figure*}

\paragraph*{Fault Injection Attack Evaluation.}
In the fault injection attack evaluation, we use the model of injecting faults to cause changes to the desired output, comparing the desired output to the one of the fault. We set the fault injection security parameter as $n=10$ for $n$-secure countermeasure~\ref{alg:fia-countermeasure} for all operations.

\begin{figure*}
    \centering
    \subfigure[Unprotected Mod. Mult.]{
        \includegraphics[width=0.23\textwidth]{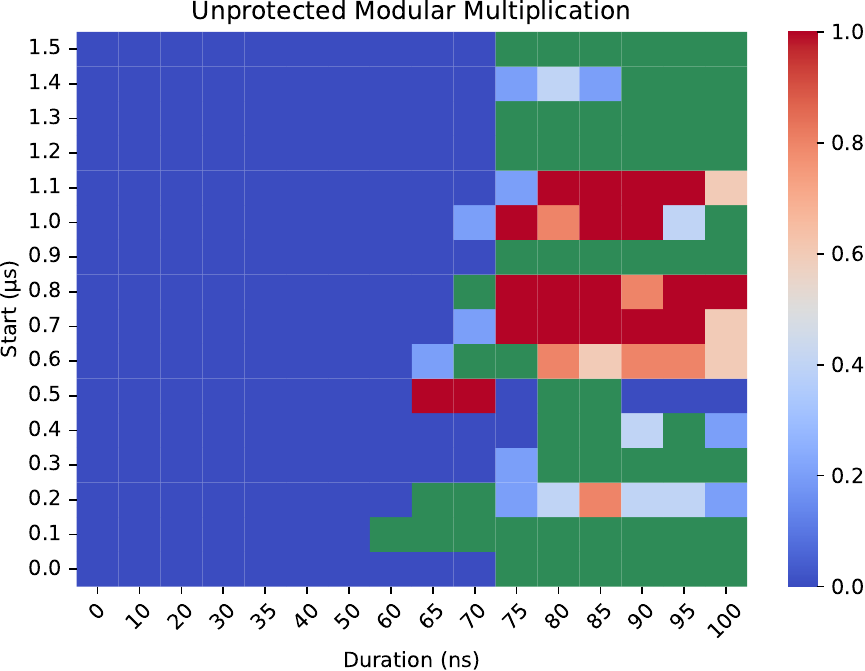}
        \label{fig:heatmaps:mod-mult-unprotected}
    }
    \subfigure[\textcolor{darkgreen}{Protected} Mod. Mult.]{
        \includegraphics[width=0.23\textwidth]{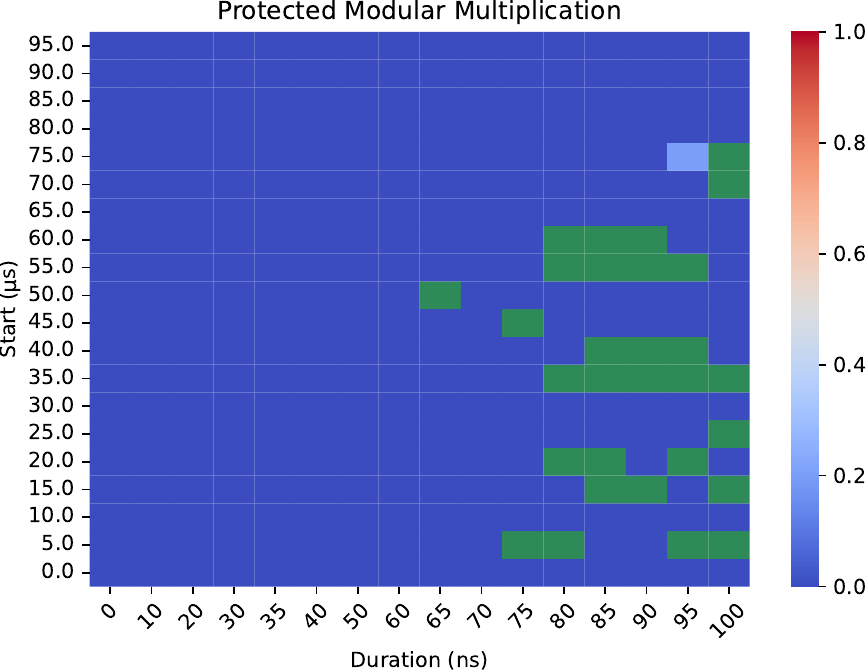}
        \label{fig:heatmaps:mod-mult-protected}
    }
    \subfigure[Unprotected Mod. Exp.]{
        \includegraphics[width=0.23\textwidth]{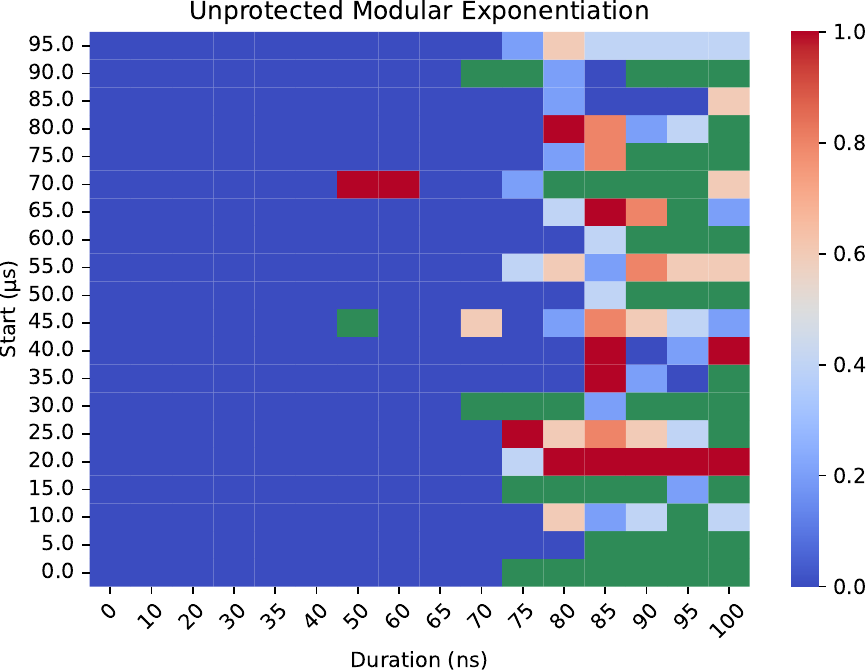}
        \label{fig:heatmaps:mod-exp-unprotected}
    }
    \subfigure[\textcolor{darkgreen}{Protected} Mod. Exp.]{
        \includegraphics[width=0.23\textwidth]{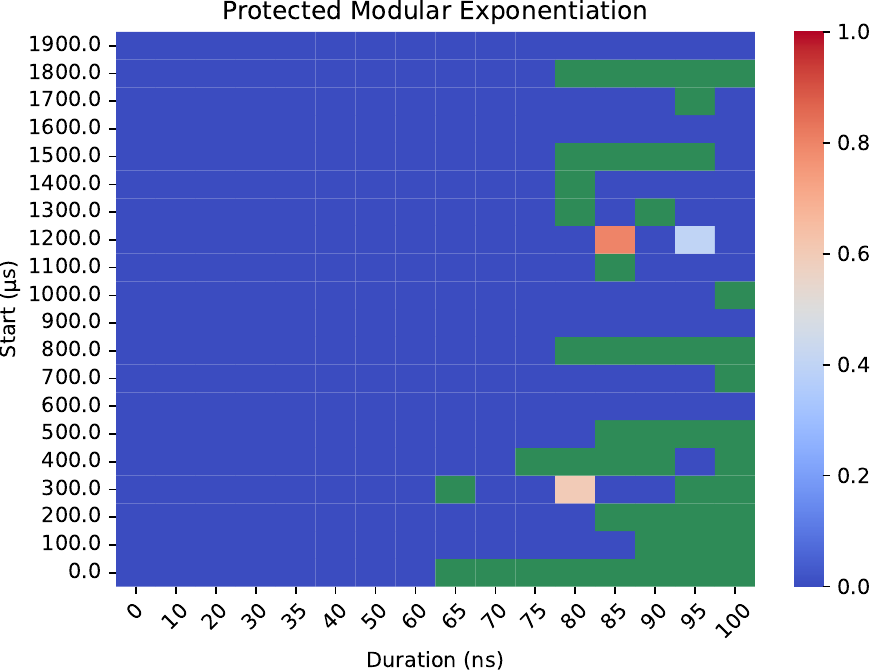}
        \label{fig:heatmaps:mod-exp-protected}
    }
    \subfigure[Unprotected Poly. Mult.]{
        \includegraphics[width=0.23\textwidth]{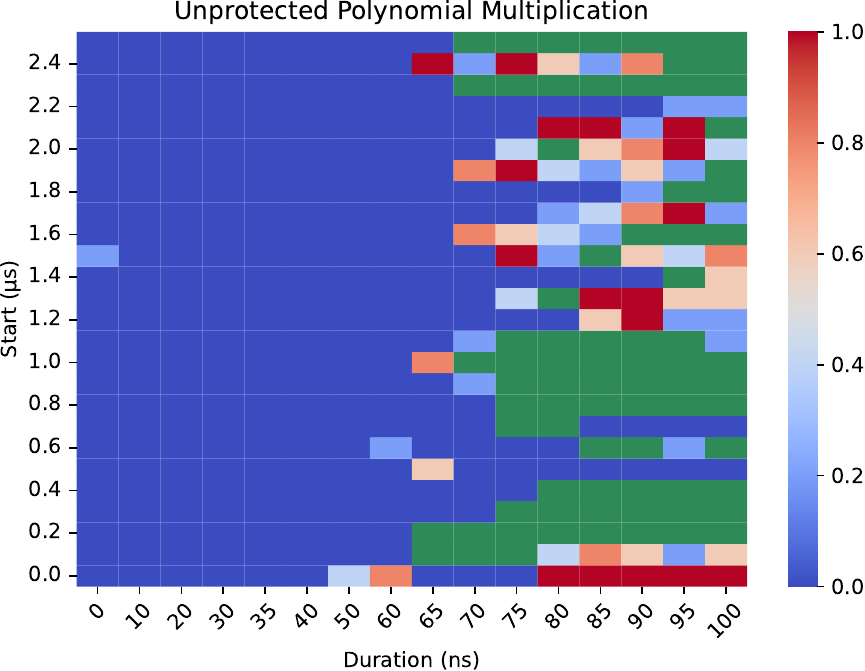}
        \label{fig:heatmaps:poly-mult-unprotected}
    }
    \subfigure[\textcolor{darkgreen}{Protected} Poly. Mult.]{
        \includegraphics[width=0.23\textwidth]{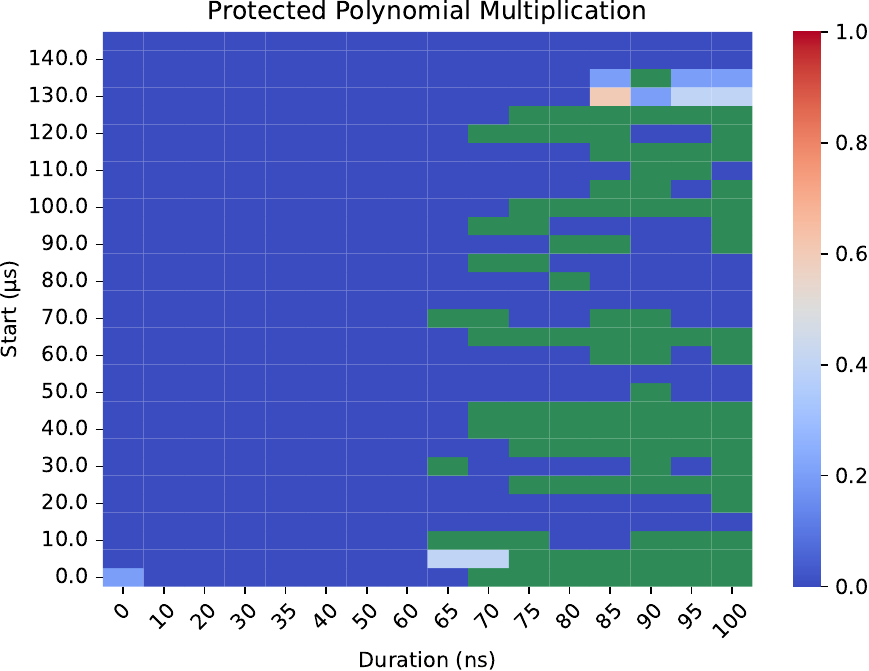}
        \label{fig:heatmaps:poly-mult-protected}
    }
    \subfigure[Unprotected NTT]{
        \includegraphics[width=0.23\textwidth]{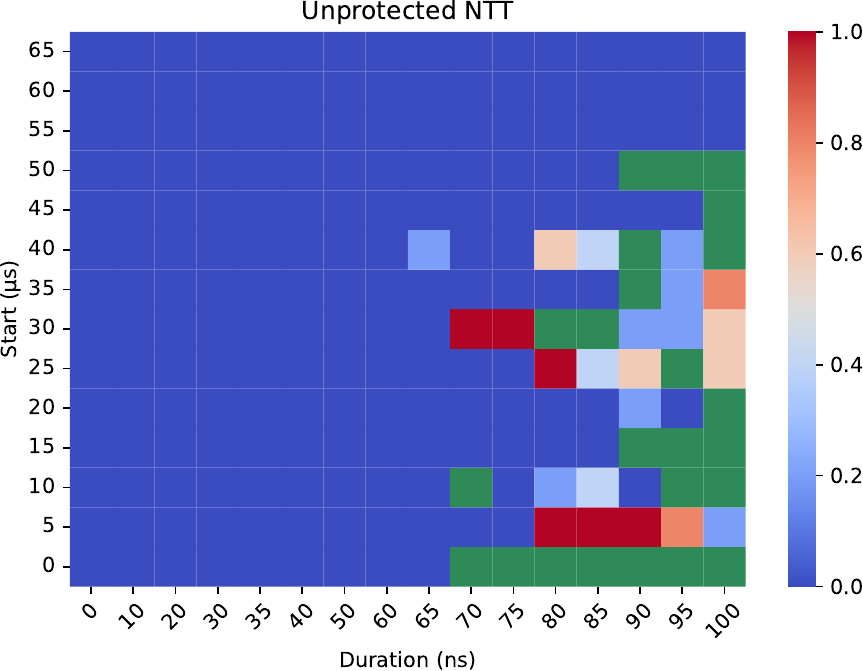}
        \label{fig:heatmaps:ntt-unprotected}
    }
    \subfigure[\textcolor{darkgreen}{Protected} NTT]{
        \includegraphics[width=0.23\textwidth]{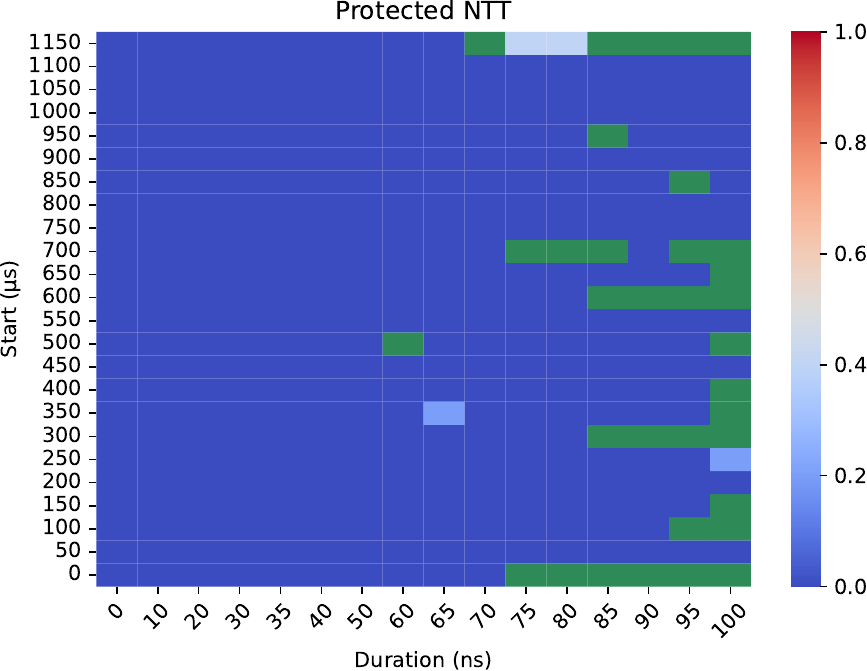}
        \label{fig:heatmaps:ntt-protected}
    }
    \subfigure[Unprotected RSA-CRT]{
        \includegraphics[width=0.23\textwidth]{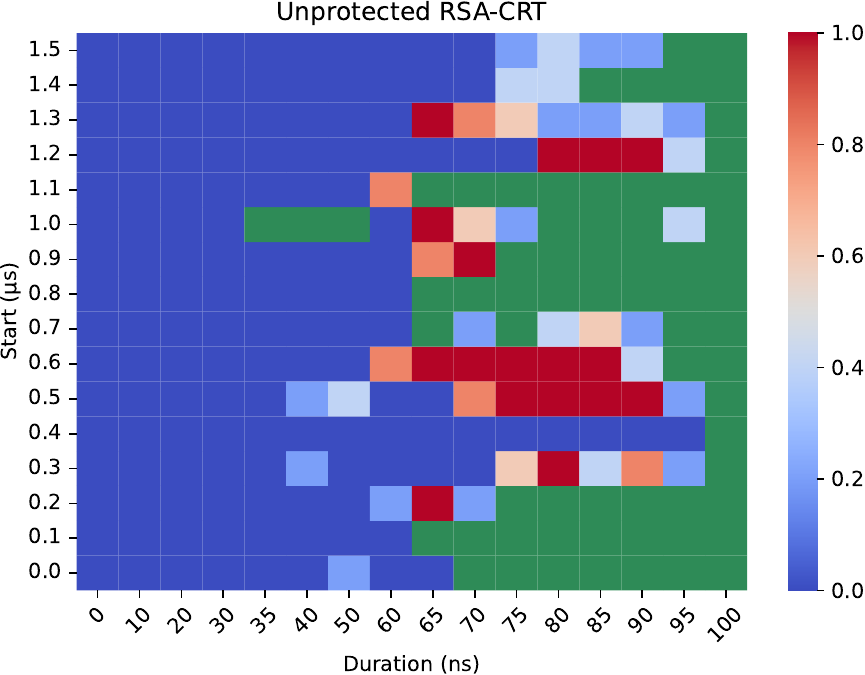}
        \label{fig:heatmaps:rsa-unprotected}
    }
    \subfigure[\textcolor{darkgreen}{Protected} RSA-CRT]{
        \includegraphics[width=0.23\textwidth]{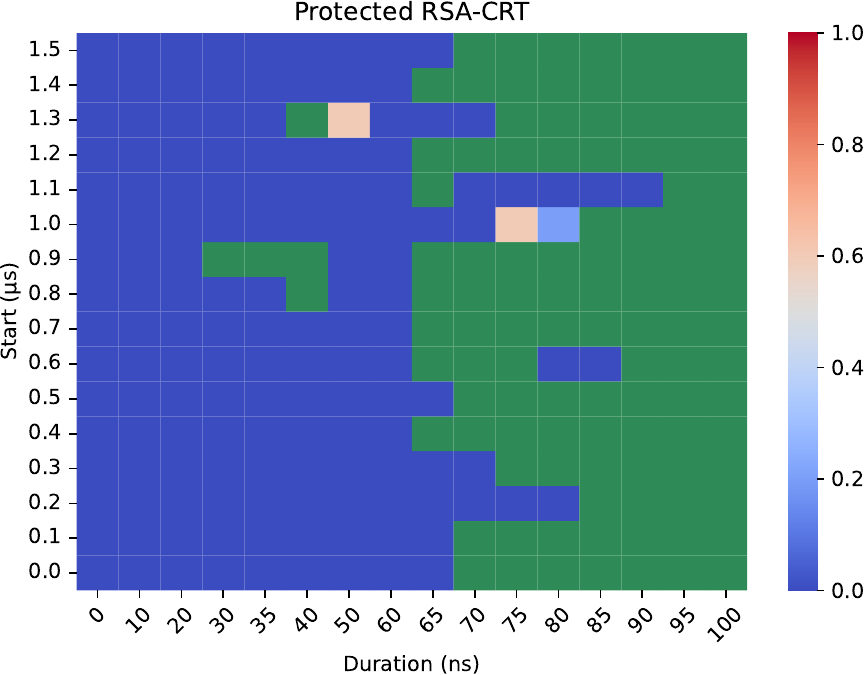}
        \label{fig:heatmaps:rsa-protected}
    }
    \subfigure[Unprotected Kyber Key Gen.]{
        \includegraphics[width=0.23\textwidth]{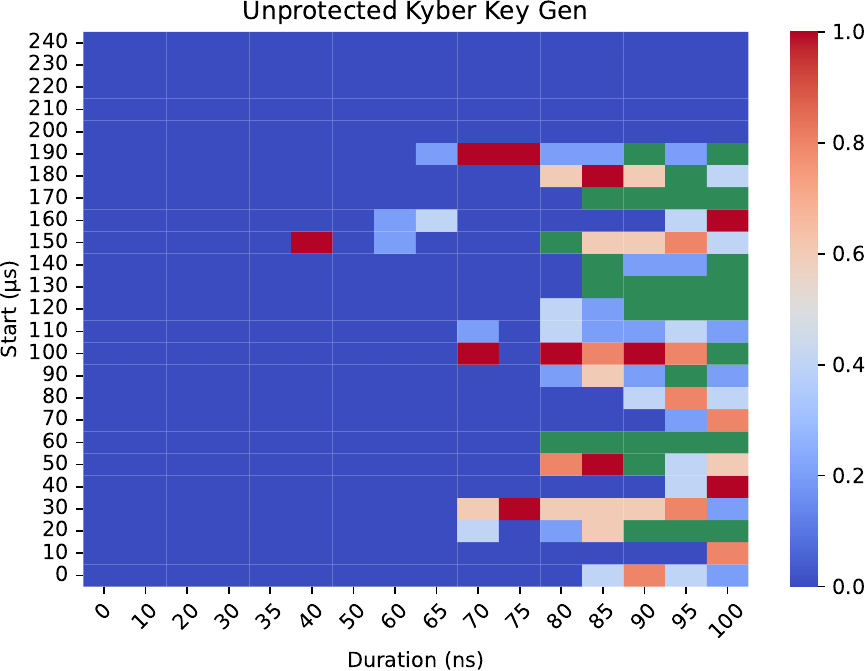}
        \label{fig:heatmaps:kyber-unprotected}
    }
    \subfigure[\textcolor{darkgreen}{Protected} Kyber Key Gen.]{
        \includegraphics[width=0.23\textwidth]{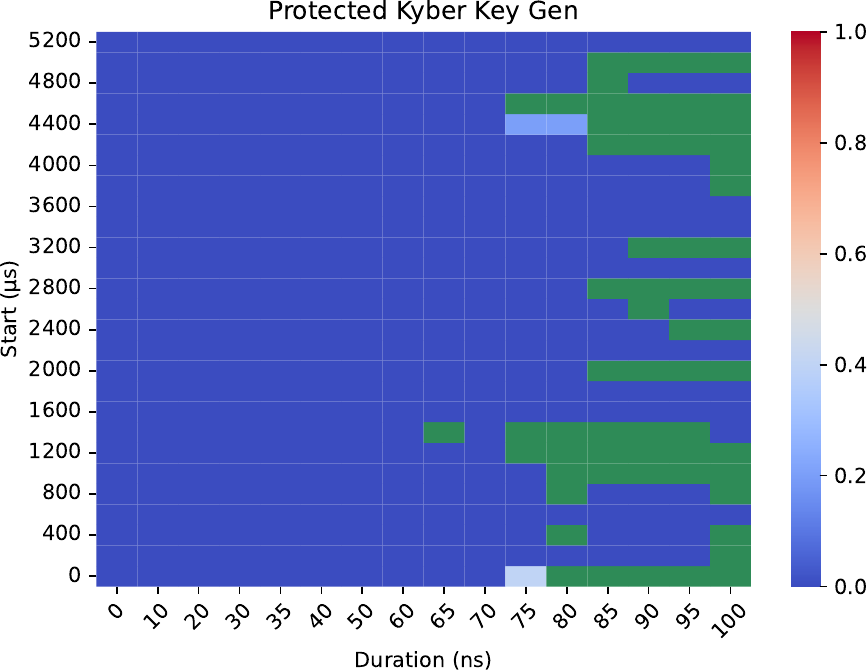}
        \label{fig:heatmaps:kyber-protected}
    }
    \caption{Fault Injection Attack Evaluation Heatmaps}
    \label{fig:fia-heatmaps}
\end{figure*}

\newcommand{\CalculatePercentage}[2]{%
  \pgfmathparse{100*(1-#2/#1)}%
  \pgfmathprintnumber{\pgfmathresult}\%}

\begin{table}[ht]
  \caption{Reduction in Faults for Different Operations}
  \label{tab:fault_reduction}
  \centering
  \footnotesize
  \begin{NiceTabular}{|Ol|Oc|Oc|Oc|}
    \toprule
    \bf Operation        & \bf Unprotected & \bf Protected & \bf Reduction                 \\
    \midrule
    Mod. exponentiation  & 165             & 9            & \CalculatePercentage{165}{9} \\
    \midrule
    Mod. multiplication  & 168             & 1             & \CalculatePercentage{168}{1}  \\
    \midrule
    NTT                  & 63              & 5             & \CalculatePercentage{63}{5}   \\
    \midrule
    Poly. multiplication & 196             & 14            & \CalculatePercentage{196}{14} \\
    \midrule
    \midrule
    RSA-CRT              & 168             & 7             & \CalculatePercentage{168}{7}  \\
    \midrule
    Kyber Key. Gen.      & 172             & 4             & \CalculatePercentage{172}{4}  \\
    \bottomrule
  \end{NiceTabular}
\end{table}

Figure \ref{fig:fia-heatmaps} presents the results of our fault attack experiment. We employed voltage glitches for the fault injection attacks. 
The Glitch Offset is the time between when the trigger is observed and when the glitch is injected. The Glitch Length is the time for which the Glitch Voltage is set. Glitch Offset and Glitch Length are the two parameters that we varied to inject faults, they correspond to the start time and duration of the glitch, respectively.
For each combination of start time and duration, we executed each target function five times, resulting in a total of 1280 test data points for each function. We used heatmaps to illustrate the ratio of faulty to correct outputs. This experiment yielded three types of outputs: faulty, correct, and board reset. In the heatmaps, colors closer to red indicate a higher likelihood of voltage glitches causing faulty outputs (red = 100\%), whereas blue signifies a lower likelihood (blue = 0\%). Green indicates instances where all test outputs resulted in the board being reset. 
We treated any output that is not the correct output as a fault, this is very conservative, as some of the outputs may not be effective faults.
From these heatmaps, the unprotected functions exhibit a significantly higher number of red dots, indicating more faults.
Furthermore, Table \ref{tab:fault_reduction} demonstrates the reduction of faults in target functions, with our protection method reducing approximately 95.4\% of faulty outputs in average, up to 99.4\% in modular multiplication. Collectively, these results affirm the effectiveness of our protection method in safeguarding the functions. 

We observed fault injection sometimes breaks memory allocations (malloc) without causing the target to crash. This is due to the fact that the target is not designed to handle such faults. We believe that this is a potential avenue for future work, as it may lead to new types of attacks. We simply reset the target in such cases, as we are not interested in the results of these attacks. However, in some cases, the fault progresses to the next operation silently without causing a crash and the target continues to operate. We registered these cases as successful attacks in the heatmaps. 

We additionally protected the fault-injection countermeasure method (Algorithm~\ref{alg:fia-countermeasure}) using classical techniques. After exiting the loop, the code verifies loop completed successfully. If not, the code resets the target. This is a simple and effective way to protect the countermeasure from fault injection attacks. This led to a reduction in faults 4.56\% in average.

\section{Limitations}\label{sec:limitations}
Our study presents a novel software-based countermeasure against physical attacks such as power side-channel and fault-injection attacks, utilizing the concept of random self-reducibility and instance hiding for number theoretic operations. While our approach offers significant advantages over traditional methods, there are several inherent limitations. Firstly, the countermeasure's effectiveness is intrinsically linked to the random self-reducibility of the function being protected. This dependency means that our approach may not be universally applicable to all cryptographic operations. Secondly, redundancy and randomness inevitably introduce computational overhead. Nevertheless, each call to original function $P$ can be easily parallelized in hardware or vectorized software implementations. This parallelization can potentially increase the noise, and we identify this as an avenue for future work. 
Finally, our approach is not tailored to defend against attacks targeting the random number generator itself. Nevertheless, there are also simple duplication based techniques to protect random number generators from physical attacks. For instance, one such technique involves comparing two successive random numbers to determine if they are identical or not, as discussed in the work of Ravi et al.~\cite{ravi2022side}.

\section{Related Work}\label{sec:related-work}

To the best of our knowledge, this work is the first to apply random self-reducibility to protect against physical attacks. However, random self-reducibility has been used in several other areas of computer science, including cryptography protocols~\cite{goldwasser2019probabilistic, blum2019generate}, average-case complexity~\cite{feigenbaum1993random}, instance hiding schemes~\cite{abadi1987hiding, beaver1990hiding, beaver1991security}, result checkers~\cite{blum1990self, lund1992algebraic} and interactive proof systems~\cite{blum1995designing, goldwasser2019knowledge, lund1992algebraic, shamir1992ip}.

Fault Injection and Side-Channel Analysis attacks are a risk for microcontrollers operating in a hostile environment where attackers have physical access to the target. These attacks can break cryptographic algorithms and recover secrets either by e.g changing the control flow of the program (FI) or by monitoring the device's power consumption with little or no evidence~\cite{cryptoeprint:2017/1082}.

Multiple countermeasures such as random delays~\cite{clavier2000differential}, masking~\cite{prouff2013masking}, infection~\cite{joye2007strengthening}, data redundancy checks~\cite{maistri2008double, medwed2008generic} and instruction redundancy~\cite{barenghi2010countermeasures} have been proposed to tackle these threats, yet their impact, effectiveness and potential interactions remain open for investigation. Therefore, our work aims to provide a new countermeasure to mitigate these attacks by combining a power side-channel countermeasure with a fault injection countermeasure~\cite{cryptoeprint:2017/1082}.

In the introduction section, we have mentioned that a fault injection attack is a kind of attack in which the attacker intentionally induced the fault to obtain the faulty output and analyze it with the original output to recover the secret, which means that a fault injection attack is based on the fault output. Therefore, if we can reduce the possibility of the attacker obtaining the faulty output, we can mitigate fault injection attacks. There are two intuitive ways to reduce the probability of the attacker obtaining the faulty output: adding the check operation, which is at the software level, or protection device, which is at the hardware level. For instance, the most famous method is Shamir's countermeasure \cite{shamir1999method}. It adds a check operation before outputting the signature to prevent the fault injection attack on RSA-CRT. Even though it has been proven that the attacker can bypass the check operation in Shamir's countermeasure and obtain fault outputs, it still provides a good concept for mitigating the fault injection attack. Some devices, such as EM pulse, voltage glitch, or laser can prevent environmental noise at the hardware level. For example, we can reduce the impact of EM pulse by adding a surge protector. We can prevent laser attacks by employing beam stops. Notice that at the hardware level, that physical device can only mitigate the particular physical fault injection method and cannot fully protect the device from all types of fault injection attacks.

\section{Conclusion}\label{sec:conclusion}
In this work, we show that if a cryptographic operation has a random self-reducible property, then it is possible to protect it against physical attacks such as power side-channel and fault-injection attacks with a configurable security. We have demonstrated the effectiveness of our method through empirical evaluation across critical cryptographic operations including modular exponentiation, modular multiplication, polynomial multiplication, and number theoretic transforms (NTT). Moreover, we have successfully showcased end-to-end implementations of our method within two public key cryptosystems: the RSA-CRT signature algorithm and the Kyber Key Generation, to show the practicality and effectiveness of our approach. 
The countermeasure reduced the power side-channel leakage by two orders of magnitude, to an acceptably secure level in TVLA analysis. For fault injection, the countermeasure reduces the number of faults to 95.4\% in average.
Although the countermeasures were introduced as software-based, they can be more efficiently implemented in hardware, particularly on FPGAs. Each call to $P$ can be parallelized in hardware, potentially increasing the noise. We identify this as an avenue for future work.

\bibliographystyle{plain}
\bibliography{IEEEabrv, bibtex/bibliography, bibtex/crypto, bibtex/pascal}

%\newpage
\appendix
\section{2PC Protocol}\label{sec:2pc}

\begin{algorithm}
  \small
  \SetKwFunction{FastGCD}{fast\_gcd}
  \SetAlgoLined
  \SetKwProg{Fn}{Function}{:}{}
  \Fn{\FastGCD{$f, g$}}{
    $d \gets \max(\text{f.nbits}(), \text{g.nbits}())$\;
    \lIf{$d < 46$}{
      $m \gets \left\lfloor\frac{{49d + 80}}{{17}}\right\rfloor$
    }
    \lElse{
      $m \gets \left\lfloor\frac{{49d + 57}}{{17}}\right\rfloor$
    }
    $precomp \gets \text{Integers}(f)\left({{f + 1}}/{{2}}\right)^{m - 1}$\;
    $v, r, \delta \gets 0, 1, 1$\;
    \For{$n \gets 0$ \KwTo $m$}{
      \If{$\delta > 0$ \textbf{and} $g \operatorname{\&} 1 = 1$}{
        $\delta, f, g, v, r \gets -\delta, g, -f, r, -v$
      }
      $g_0 \gets g \operatorname{\&} 1$\;
      $\delta, g, r \gets 1 + \delta, \frac{{g + g_0 \cdot f}}{2}, \frac{{r + g_0 \cdot v}}{2}$\;
      $g \gets \text{ZZ}(g)$\;
    }
    % $v_{\text{out}} \gets \text{sign}(f) \cdot \text{ZZ}(v \cdot 2^{m - 1})$\;
    $inverse \gets \text{ZZ}(\text{sign}(f) \cdot \text{ZZ}(v \cdot 2^{m - 1}) \cdot \text{precomp})$\;
    \KwRet{$inverse$}%\;
  }
  \caption{Fast GCD ($f, g$)}\label{alg:fast-gcd}
\end{algorithm}

\begin{algorithm}
  \small
  \SetKwFunction{ModularExponentiation}{modular\_exponentiation}
  \SetAlgoLined
  \SetKwProg{Fn}{Function}{:}{}
  \Fn{\ModularExponentiation{$x, y, p$}}{
    $res \leftarrow 1$\;
    $x \leftarrow x \mod p$\;
    \If{$x = 0$}{
      \KwRet{$0$}\;
    }
    \While{$y > 0$}{
      \If{$(y \operatorname{\&} 1) = 1$}{
        $res \leftarrow (res \cdot x) \mod p$\;
      }
      $y \leftarrow y \gg 1$\;
      $x \leftarrow (x \cdot x) \mod p$\;
    }
    \KwRet{$res$}\;
  }
  \caption{Modular Exponentiation ($x, y, p$)}\label{alg:modular-exponentiation}
\end{algorithm}

\end{document}